\documentclass[11pt]{article}
\usepackage[margin=1in]{geometry}
\usepackage[small, bf, margin=28pt, tableposition=top]{caption}
\usepackage{authblk}

\usepackage{multirow}
\usepackage{rotating} 
\usepackage{subfigure}
\usepackage[all]{xy}
\usepackage{amsthm}
\usepackage{amsmath}
\usepackage{amssymb}
\usepackage{psfrag}
\usepackage[displaymath]{lineno}

\newcommand{\At}{\textbf{A}}
\newcommand{\Pt}{\textbf{P}}
\newcommand{\Ft}{\textbf{F}}
\newcommand{\Tt}{\textbf{T}}
\newcommand{\Gt}{\textbf{G}}
\newcommand{\Dt}{\textbf{D}}
\newcommand{\Yt}{\textbf{Y}}

\newcommand{\E}{\mathbf{E}} 
\newcommand{\W}{\phi} 

\theoremstyle{plain}
\newtheorem{theorem}{Theorem}

\newtheorem{lemma}[theorem]{Lemma}

\theoremstyle{definition}

\theoremstyle{remark}
\newtheorem{remark}{Remark}

\begin{document}
\title{The Iterated Prisoner's Dilemma on a Cycle}

\author{Martin Dyer}
\author{ Velumailum Mohanaraj}
\affil{ School of Computing, University of Leeds \\ Leeds LS2 9AS, United Kingdom}

\date{}

\maketitle

\abstract{Pavlov, a well-known strategy in game theory, has been shown to have
some advantages in the Iterated Prisoner's Dilemma (IPD) game. However, this strategy
can be exploited by inveterate defectors. We modify this strategy to mitigate the
exploitation. We call the resulting strategy Rational Pavlov. This has a
parameter $p$ which measures the ``degree of forgiveness'' of the players. We
study the evolution of cooperation in the IPD game, when $n$ players are arranged
in a cycle, and all play this strategy. We examine the effect of varying $p$
on the convergence rate and prove that the
convergence rate is fast, $O(n \log n)$ time, for high values of $p$. We also prove that the convergence rate is exponentially slow in $n$ for small
enough $p$. Our analysis leaves a gap in the range of $p$, but simulations suggest that there is, in fact, a sharp phase transition.}


\section{Introduction}
\label{s:introduction}
\subsection{Overview}

The Prisoner's Dilemma (PD) is one of the most famous strategic games in game theory (see, for example, \cite{nisan2007}). This game is widely used
as a prototype for the study of the evolution of cooperation among selfish agents. It has attracted a large amount of interest from researchers in
diverse fields, due to the fact that it represents a very common strategic situation that needs to be understood. Many
real-life problems can be modelled by the Prisoner's Dilemma~\cite{brembs-1996a}.

In the standard form of PD, the payoff obtained when both prisoners cooperate with each other is denoted by $R$, the \emph{reward for mutual
cooperation}. The payoff gained when both defect is denoted by  $P$, the \emph{punishment for mutual defection}. Finally, $T$
(the \emph{temptation to defect})  is earned by the informer and $S$ (the \emph{sucker's payoff}) is earned by the other when one defects
and the other cooperates. These four outcomes are shown in Figure~\ref{tbl:payoff} in a matrix form.  In this game, the payoff of an action does not depend on the player. Hence the game is said to be \emph{symmetric}.

\begin{figure}[ht!]
\centering
\begin{tabular}{cc|c|c|}
\multirow{4}{10mm}{\begin{sideways}\parbox{15mm}{
\textbf{Row Player}}\end{sideways}}& \multicolumn{3}{c}{\textbf{Column Player}} \\
 & \multicolumn{1}{c}{\ }& \multicolumn{1}{c}{Cooperate} & \multicolumn{1}{c}{Defect} \\
\cline{3-4}
& Cooperate & $R=3$ & $S=0$ \\
\cline{3-4} & Defect & $T=5$ & $P=1$\\ \cline{3-4}
\end{tabular}
\caption[The payoff matrix for Prisoner's
Dilemma]{\label{tbl:payoff}The payoff matrix for the Prisoner's
Dilemma with Axelrod's numerical example. The game is symmetric, therefore only the payoffs to the row player are shown. }
\end{figure}

The interesting feature of the PD game is the property that the four payoffs satisfy: $T>R>P>S$.
Hence, in one shot game, it is always best to defect. Thus, self-interest of the players leads to
the payoff $P$ which is worse than the $R$ that both players would get by cooperating, hence the dilemma.
In the Iterated Prisoner's Dilemma (IPD), the same players meet again with a high probability, thus getting an opportunity to punish
each other for any previous non-cooperative moves.  Understandably, the fear of retaliation here is likely to encourage cooperation.
This was studied in \cite{axelrod:cooperation}, which stimulated work in this area.
Another constraint $2R > S+T$ is usually added to the standard form of the IPD~\cite{may1987}. If this constraint is not present, players could benefit from receiving $S$ and $T$ on alternate rounds rather
than $R$ on every round through continuous cooperation.

A great deal of research has been done to find out an ideal strategy for the IPD.
A strategy helps  players decide whether to cooperate or defect in the current
round. A simple strategy called \emph{tit-for-tat} (TFT)
surprisingly won Axelrod's seminal computer tournaments~\cite{axelrod:cooperation}. TFT cooperates on the first round, and
copies what the opponent has done on the previous round thereafter. However, this strategy has two main problems:  firstly,
it is not \emph{evolutionarily stable}~\cite{robert1987,imhof2007574}; and secondly, any mistakes by the agents or any noise in the responses
may cause a misinterpretation leading to irrecoverable retaliation sequences. (Informally, a strategy is evolutionary stable if
a population of players adopting the strategy can not be overrun by any mutant strategy.)

Another well-known strategy is called \emph{Pavlov}. The Pavlov, an exemplar of the \emph{win-stay lose-shift}
strategy,  works as follows. On each iteration of the game, if a Pavlov player's payoff is one of the two smaller payoffs, i.e.\ $P$ or $S$, then he switches his action in the next round of the game, otherwise he keeps the same action.
It is claimed in \cite{imhof2007574, kraines-1993a,nowak-1993a} that the Pavlov performs  better than TFT.  This is due to its ability to recover
from noise and errors and its capability to exploit unconditional cooperators (All-C). However Pavlov has two main issues. Firstly, Pavlov is deterministic, thus it cannot represent uncertainties present in the real world, such as the stochastic nature of biological interactions~\cite{martinsigmund1990}.
Secondly, it fares poorly against all-time defectors (All-D). This is because, when played against All-D, Pavlov is punished for defecting,
so switches to cooperation, just to be punished even more. This is repeated forever, and consequently Pavlov collects the sucker's payoff
($S$) on alternate rounds.

A family of stochastic Pavlovian strategies $\mathcal{P}(k,\ell)$, for a fixed $\ell$ and $ 0 < k <\ell$,
has also been studied and hailed as a near-ideal strategy for the IPD in~\cite{kraines-1993a}.
$\mathcal{P}(k,\ell)$ cooperates with probability $k/\ell$. At the end of each round of the game,
$k$ is increased if the player gains $T$ or $R$, and decreased otherwise.
The advantages of these strategies are: they are adaptive and naturally stochastic.
The disadvantages are:  they take exponential time in $\ell$
for learning to cooperate and are exploitable by All-D. It is worth to mention
that  $\mathcal{P}(1,1)$
is equivalent to the Pavlov strategy described above.

Before we move on, let us represent the Pavlov strategy as a (deterministic) Markov chain.
Suppose two agents play the IPD using Pavlov. This can be
modelled as a Markov chain having four states, each representing a
possible combination of the strategies of the agents. We denote
these states by $\mathbf{++}$,  $\mathbf{+-}$,  $\mathbf{-+}$ and
$\mathbf{--}$. (Here $+$ stands for \emph{cooperation} and $-$ stands
for \emph{defection}.) Thus, $++$, for example, represents the
scenario where both agents  cooperate. The transition diagram for this
process is shown in Figure~\ref{fig:ps}.

\subsection{Rational Pavlov strategies on IPD}

It is now clear that the main weakness of the Pavlov is that it can be exploited by All-D.
Thus we suggest an enhancement to this strategy. We modify it to add randomness.
This, we think, makes the resultant strategy more rational and robust. The details
of the modification are given below.

A Pavlov player cooperates in the current round if both he and his opponent  cooperated or  defected at their previous play. Thus, a
transition from  $--$ to $++$ happens in a single repetition with certainty. We will modify this in two ways, so that the transition
from $--$ to $++$ will only happen with some probability less than 1. More precisely, the modifications introduced to the $--$ to $++$
transition are: if both players defected , i.e.\ in state $--$, in the previous play, then

\begin{enumerate}
\item  each player decides independently whether to cooperate in the current round with probability $p$. The transition diagram of the
strategy obtained after this modification is shown in Figure~\ref{fig:rps}. As we believe that this modification adds some rationality
to Pavlov we call this strategy  Rational Pavlov (RP).
\item both cooperate in the current round with probability $p$. The transition diagram of this strategy is shown in Figure~\ref{fig:srps}. This is a
simplified version of the RP, hence the name Simplified Rational Pavlov (SRP). Even with the absence of communication, players deciding
together with probability $p$ can also be justified using the superrationality  principle~\cite{metamagicaltheme96}. Thus, SRP
might also be expanded as Super Rational Pavlov.

It is noteworthy that both RP and SRP are equivalent when $p=1$ or $p=0$. And, both RP and SRP
reduce to the original Pavlov strategy when $p=1$.
\end{enumerate}

\begin{figure}[ht]
\centering
\subfigure[Pavlov Strategy (PS)]{
\label{fig:ps}
\xymatrix{
&*+[o][F-] {++}\ar@(ul,ur)[]^{1} & \\
*+[o][F-]{-+}\ar[dr]^{1} & & *+[o][F-]{+-}  \ar[dl]_{1} \\
&  *+[o][F-]{--} \ar[uu]_{1} & \\ & \\ }
}
\subfigure[Rational Pavlov (RP)]{
\label{fig:rps}
\xymatrix{
&*+[o][F-] {++}\ar@(ul,ur)[]^{1} & \\
*+[o][F-]{-+}\ar@/^/[dr]^{1} & & *+[o][F-]{+-}  \ar@/_/[dl]_{1} \\
&  *+[o][F-]{--} \ar@(dl,dr)[]_{(1-p)^2}\ar@/_/[ur]_{p(1-p)} \ar@/^/[ul]^{p(1-p)} \ar[uu]_{p^2} &  }
}
\subfigure[Simplified RP (SRP)]{
\label{fig:srps}
\xymatrix{
&*+[o][F-] {++}\ar@(ul,ur)[]^{1} & \\
*+[o][F-]{-+}\ar[dr]^{1} & & *+[o][F-]{+-}  \ar[dl]_{1} \\
&  *+[o][F-]{--} \ar@(dl,dr)[]_{(1-p)}\ar[uu]_{p} &  }
}
\caption[Transition diagrams for RP and SRP]{Transition diagrams of the original and the modified Pavlovian strategies. Here, ``$-$'' represents cooperation and ``$+$'' represents defection. The transition probabilities are shown on the edges. }
\label{fig:statesDiagram}
\end{figure}
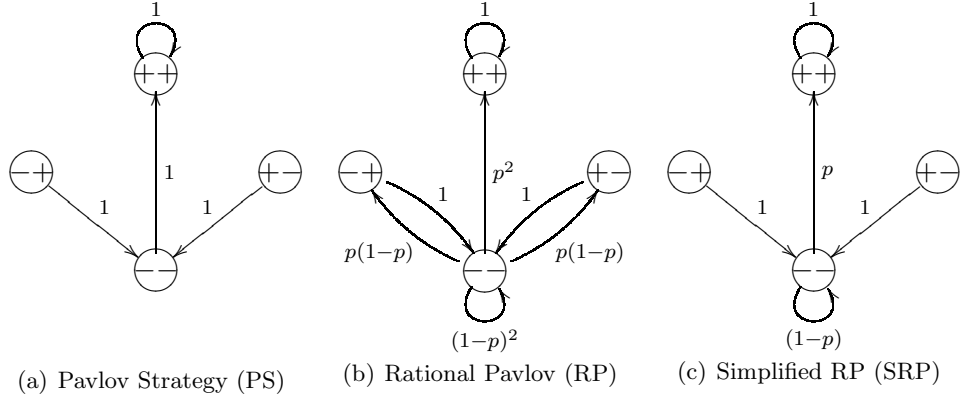

\subsection{Previous work}\label{s:previouswork}

Although our work appears to be the first to formally define strategies
like RP and SRP, there is some evidence in the literature that support
our intuition behind the proposed improvements.
Firstly, the results from experiments with humans in~\cite{CWedekind04021996} overwhelmingly
support this. The results show that humans use a Pavlov-like strategy that is smarter than the classic Pavlov strategy when dealing with All-D. This Pavlov-like strategy cooperates after $--$ state with probability
less than 1,  like RP and SRP do. Not surprisingly, the players using this strategy were more successful than
the others in the experiments. Furthermore, a similar modification has been suggested as a possible improvement
of the Pavlov in~\cite{kraines-1993a}. Finally, a strategy similar to RP and SRP has proved
to be the winner in computer simulations as well~\cite{frean94} .

Apart from reigning supreme in evolutionary game theory,
the Pavlov has been studied in distributed Artificial Intelligence
as a learning model. Shoham and
Tennenholtz~\cite{shoham93colearning} introduced the notion of
\emph{co-learning} where agents try to adapt
to their environment by adapting to one another's behaviour. In the
same paper, they also defined a simple co-learning update rule,
namely \emph{Highest Cumulative Reward} (HCR). This rule states
that an agent should adapt to the action that resulted in favourable
feedback in the latest $\mu$ iterations, where $\mu$ is the memory
size. The HCR update rule ensures that cooperation emerges at the end
in the IPD game. This update rule with $\mu =0$, which is one of the most efficient
memory sizes~\cite{kittock93emergent}, is precisely the Pavlov strategy.

Shoham and Tennenholtz~\cite{shoham93colearning} studied the evolution of cooperation for the HCR
update rule in unstructured population and concluded that it is an impractical  model for the
evolution of cooperation. This conclusion is not surprising, as it is now well known that, in an unstructured
population, natural selection favours defection over cooperation~\cite{Ohtsuki2006}.
Hence, there is a growing interest in studying the evolution of cooperation when the
topology for interactions is not complete (see, for example, \cite{santos2006,Ohtsuki2006}).
Thus we consider the players to be arranged as the vertices of a graph, and they can interact
only along the edges of the graph.
Kittock~\cite{kittock93emergent} studied the effects of an interaction graph on the emergence
of cooperation under the dynamics in which, at every step, two adjacent players
are selected uniformly at random to play the IPD game using the Pavlov strategy. The paper~\cite{kittock93emergent}  presented the
results from an empirical study which shows that the time needed for the emergence of
cooperation in the IPD game is polynomial on cycles and exponential on complete graphs.

Most of the work we have mentioned above is empirical. However the need for rigorous results has rightly
been emphasised. (See, for example, \cite{kittock93emergent,may1987,shoham93colearning}.)
The reason for the lack of rigorous analysis of games on graphs is that it is complicated
due to the vast number of patterns that can be generated~\cite{Nowak2006}.
While the empirical results do give some insights into the evolution, some
of the results are far too complicated to be understood without theoretical
backing. The results obtained
through rigorous analysis are  often more revealing and contribute to a
clearer understanding of the problem . Hence, in this paper, we analyse the behaviour of RP rigorously. More precisely, we establish the conditions
for fast convergence, and determine the rate of convergence to cooperation
when all players play RP. These measures are central to understanding the emergence of cooperation among selfish agents~\cite{robert1987}.

On the theoretical side, Dyer \emph{et al.}~\cite{dyer00convergence}
studied the two cases examined in~\cite{kittock93emergent}, using rigorous analysis.  Mossel and Roch~\cite{mossel2007} did a
similar study for some expander graphs and bounded degree trees and showed that
the convergence is slow in both settings for the Pavlov strategy.
Istrate \emph{et al.}~\cite{Istrate2008} investigated the robustness
of these convergence results under \emph{adversarial scheduling} in which  an adversary selects which players update at every step. Their results show
that if an adversary can specify two players for the update, the game might never
converge. Along this line of work,  we carry out a rigorous analysis of RP in this paper.  In particular, we attempt to find the range of
$p$ that favours fast evolution of cooperation and the range of $p$ that makes the evolution of cooperation exponentially slow, when the IPD is played on the cycle using RP.
(Here, we consider speed of convergence as a function of the number of players, $n$.)  All our results are complemented by simulation results. Our choice of graph, the cycle, is an extreme case, where every player has only two
neighbours. Game dynamics have  previously been analysed for the cycle~\cite{dyer00convergence,kittock93emergent,Ohtsuki2006}.
Our results show some interesting results, for instance, we show that the emergence of convergence is exponentially slow for small values of $p$. Thus a high degree of forgiveness seems necessary for cooperation to emerge. Perhaps, our most important message is that a Rational Pavlov player can reduce
the risk of being exploited without compromising the emergence of cooperation.

We have analysed SRP as well. The analysis is quite similar to that of RP. Therefore,
we do not present it in this paper. Instead, we make some
remarks on the final results under relevant sections.

\subsection{Preliminaries} \label{s:preliminaries}

Much of the notation and terminology used in this paper is adopted from \cite{dyer00convergence}.
We consider $n$ players arranged as the vertices of a cycle graph $G=(V,E)$, where $V = \{0,\ldots, n-1\}$ and
$E = \{\{i, i+1\} \,:\, 0 \leq i \le n-1\}$. Hence, vertex $i$ can interact only with the vertices
$i-1$ and $i+1$. Here and throughout the paper, addition and subtraction on vertices is performed modulo $n$.

The agent at the vertex $i$ $(0  \leq   i  <  n)$ has a state $S_i \in \{-1, 1\}$, where $-1$ represents \textit{defection}
and $1$ represents \textit{cooperation}. We will denote the cooperator-states, or $1$'s, also as $+$'s (pluses), and
the defector-states, or $-1$'s, as $-$'s (minuses). Each edge of the graph has a state which is determined by the states
of its end vertices. Thus, an edge of the graph might be in any of four states, $--,-+,+-,++$, as shown in the state transition
diagrams in Figure~\ref{fig:statesDiagram}.

In this study, the game is played in the following way. At each stage, an edge of the cycle is selected uniformly at random.
The agents connected by this edge play the game using RP and update their strategies accordingly. In this process, emergence of
cooperation means reaching the state where everyone cooperates, in other words, reaching the state $S^*$ with $S_i = 1$
for all $ i \in V$. The state $S^*$ is the unique absorbing state of this process.

We will use the following terminology. Let $S \in \{-1, 1\}^V$ be given. A \emph{plus-run} (resp. minus-run) in $S$ is
an interval $[i, j]$  where $0 \leq i, j < n$, such that $S_k = 1$ (resp. $-1$) for $i\leq k \leq j$ and $S_{i - 1} = -1$
(resp. 1), $S_{j + 1} = -1$ (resp. 1). (It is possible to have $j \le i$, since we are working modulo $n$.) Clearly all runs are disjoint. The length of a minus-run $R_\textrm{d}$, denoted by $\ell(R_\textrm{d})$, equals the number of minuses in the run.
We will refer to a minus-run of length $\ell$ as an $\ell_\textrm{d}$-run where the subscript ``d'' stands for defectors.
We use similar variables for a plus-run with subscript ``c'', which stands for cooperators.

We now give some definitions for minus-runs, which are equally applicable to plus-runs if the signs
are changed, and the subscript c is used. A $1_\textrm{d}$-run is also called a \emph{singleton} minus,
and a $2_\textrm{d}$-run is also called a \emph{pair} of minuses. There are two \emph{outer rim} edges associated with
a minus-run $R_\textrm{d} = [i,j]$, namely $\{i - 1, i\}$ and $\{j, j+1\}$. The all-minuses configuration is not a run as we have defined it,
since it has no bordering pluses, we will nevertheless refer to it as the $n_\textrm{d}$-run.

Finally, the parameter of both RP and SRP will be denoted by $p$, but the context should always make the meaning clear.
The following theorems summarise our results.
\medskip
\begin{theorem}
\label{thm:fastconvRP}
Suppose $n$ players, arranged as the vertices of a cycle, play the
\emph{IPD} game using Rational Pavlov \emph{(RP)} with
parameter $p \ge 0.870$. Then, there is a constant $\omega>0$ such that,
the probability that the all-cooperate
state is not reached in time
\[\dfrac{1}{\omega}n \log \left( \dfrac{n}{\varepsilon} \right)\]
is at most $\varepsilon$, for any $\varepsilon>0$.
\end{theorem}
\medskip
\begin{theorem}
\label{thm:slowConRP}
Suppose $n$ players, arranged as the vertices of a cycle, play the
\emph{IPD} game using Rational Pavlov \emph{(RP)} with
parameter $p$. Suppose all players play \emph{defect} when the game is started.
Then there exists a constant $p_1 > 0$ such that, for all $p \leq p_1$, it takes time
exponential in $n$ for the all-cooperate state to be reached, except for
 probability exponentially small in $n$.
\end{theorem}
\medskip
\begin{theorem}
\label{thm:defection}
Suppose $n$ players, arranged as the vertices of a cycle, play the
\emph{IPD} game using Rational Pavlov \emph{(RP)} with $p=0$. Provided there is at least one defector on the cycle at
the beginning of the game, the game converges to defection in time $T_n$ where $T_n$ lies within the range $\Bigl[ \frac{n(n-1)}{2} \pm   O\bigl(n^{\frac{3}{2}}\log n\bigr) \Bigr]$ with high probability.
\end{theorem}
\medskip
\begin{remark}
In this paper, an event $Y_n$ which depends on the size of the graph $n$ is said to happen
\emph{ with high probability}, or in short \emph{w.h.p.}, to mean that $\Pr(Y_n) \to 1$ as $n \to \infty $.
\end{remark}
\medskip
The outline of the rest of this paper is as follows: In Section~\ref{s:fastconvergence}, we derive the conditions for fast convergence to cooperation
when RP is used on the cycle.  In Section~\ref{s:slowconvergence}, we prove that convergence to cooperation is slow for small
values of $p$. Section~\ref{s:defection}
concentrates on a special case where defection emerges fast on the cycle. Experimental results are presented in Section~\ref{s:simulation}. Finally,  Section~\ref{s:conclusion} presents our concluding remarks.

\section{Fast convergence on the cycle}
\label{s:fastconvergence}

The convergence rate of the IPD has been analysed in \cite{dyer00convergence} by finding a nonnegative integer-valued potential
function $\xi:\{-1, 1\}^V \rightarrow \mathbb{R}$ such that $\xi(S) = 0$ when $S = S^*$ and $\xi(S) > 0$  otherwise.
Then, Dyer \emph{et al.}~\cite{dyer00convergence} proved that the expectation of the function $\xi$,  which measures the distance from the absorbing state $S^*$ to any given state $S$, decreases with
non-null probability till the absorbing state is reached. We use a similar approach here, but with a simpler potential
function $\W(S)$. This function is defined as

\begin{equation}
\label{e:potentialfunction}
\begin{aligned}
    \W(S) & = \sum_{\ell=1}^n w_{\ell}r_{\ell} ,
\end{aligned}
\quad\text{where}\quad
 \begin{aligned}
    & r_{\ell} \text{ is the number of $\ell_\textrm{d}$-runs, and }  \\
    & w_{\ell} > 0 \text{ is the weight of an $\ell_\textrm{d}$-run.}
\end{aligned}
\end{equation}

Note that $\W(S)  = 0 \text{ when } S=S^{*} \text{ and } \W(S)  > 0$ otherwise.
In this section, we prove Theorem~\ref{thm:fastconvRP} which shows that the emergence of cooperation is fast when using RP in the IPD for high values of $p$. This is done by
studying the changes in the total weight of the minus-runs. Hence, in this section, a run means a run of minuses unless otherwise stated.

\subsection{Analysis}
\label{s:RPSProbFormulate}

We first consider the minus-runs that are separated from their adjacent runs by at least two pluses. When two minus-runs are separated
by a singleton plus, choosing the outer rim edges of the singleton plus causes the two runs to
merge together. This case therefore needs some special consideration and is addressed at  the end of this section.

We need to show that the expectation of $\W$ decreases after every iteration of the game. This requirement can be modelled by having a constraint that the expected total weight of the runs  created by hitting an overlapping
edge of an $\ell_\textrm{d}$-run ($\ell =1,2,\ldots,n$), denoted  by $\E [s_{\ell}]$, is strictly less than the original weight $w_\ell$.
We will now consider runs of different lengths in turn, and find the corresponding constraint.

\smallskip
\noindent
\textbf{A $1_\textrm{d}$-run.}
For a $1_\textrm{d}$-run, there are only two edges which overlap this run. Choosing either of these edges will produce
a $2_\textrm{d}$-run. Therefore, the $1_\textrm{d}$-run can be handled by adding the following constraint to the formulation.
\begin{equation*}
\E[s_1] = \frac{1}{n} \big( 2 w_2 + (n-2) w_1 \big) \le (1-\delta)w_1, \\
\end{equation*}
for small $\delta >0$. Let $\delta = \omega /n$. Thus we obtain
\begin{equation}
\label{e:singleton}
2  w_2 - (2-\omega)  w_1 \le 0\ .
\end{equation}

\smallskip
\noindent
\textbf{An $\ell_\textrm{d}$-run, where $ 2 \leq \ell \leq n-1$.}
There are $\ell+1$ edges which overlap this run. Two of them are outer rim
edges, and selecting either for the update causes the run to grow in length by 1.
All other $\ell-1$ overlapping edges are in state $--$. Let us number these
edges $1,2, \ldots, \ell-1$. According to the strategy RP, if the edge $i \in \{1,2, \ldots, \ell-1\}$
is chosen for the play, this edge will become $++$ with probability $p^2$, producing a ($i-1, \ell-i-1$)-split.
Similarly, this edge might go to the state $+-$ or $-+$ with a probability of $p(1-p)$, resulting in a $(i-1, \ell-i)$-split
or a $(i, \ell-i-1)$-split respectively. The edge might also remain in the same state with probability $(1-p)^2$.
Finally, there is a chance of not hitting any of the overlapping edges of the run, leaving the $\ell_\textrm{d}$-run intact.
We can now compute the expected new weight of the run after one step of the game, by combining these cases.
Hence we have
\begin{equation*}
\begin{split}
&\E[s_{\ell}] =  \frac{1}{n} \biggl( 2w_{\ell+1} + p^2  \sum_{i=1}^{\ell-1}(w_{i-1} + w_{\ell-i-1}) + p(1-p) \sum_{i=1}^{\ell-1}(w_{i-1} + w_{\ell-i}) + \\
& p(1-p)\sum_{i=1}^{\ell-1}(w_i + w_{\ell-i-1}) + (1-p)^2 (\ell-1)w_{\ell} \biggr) + \frac{n-(\ell+1)}{n} w_{\ell} \le (1-\delta)w_{\ell}\ .
\end{split}
\end{equation*}
This inequality can be simplified to
\begin{equation*}
\begin{split}
2w_{\ell+1} + p^2  \sum_{i=1}^{\ell-1}(w_{i-1} + w_{\ell-i-1}) + p(1-p) \sum_{i=1}^{\ell-1}(w_{i-1} + w_{\ell-i}) & + \\
p(1-p)\sum_{i=1}^{\ell-1}(w_i + w_{\ell-i-1}) + (1-p)^2 (\ell-1) w_{\ell} &\le (\ell+1-\omega)w_{\ell}\ .
\end{split}
\end{equation*}
Hence, we have
\begin{equation*}
2w_{\ell+1} + 2p^2  \sum_{i=0}^{\ell-2}w_i + 2p(1-p) \biggl( \sum_{i=0}^{\ell-2}w_i + \sum_{i=1}^{\ell-1}w_i \biggr) + (1-p)^2 (\ell-1) w_{\ell} \le (\ell+1-\omega)w_{\ell}\ .
\end{equation*}
Thus,
\begin{equation*}
\begin{split}
2w_{\ell+1} + 2p^2  \sum_{i=0}^{\ell-2}w_i + 2p(1-p) \biggl(2\sum_{i=0}^{\ell-2}w_i -w_0 + w_{\ell-1} \biggr) &+ \\
(1-p)^2 (\ell-1)w_{\ell} &\le (\ell+1-\omega)w_{\ell}\ .
\end{split}
\end{equation*}
Let $w_0 = 0$. Then, for $ 2 \leq \ell \leq n-1$, we have
\begin{equation}
\label{e:shortequation}
2w_{\ell+1} + 2p(2-p)  \sum_{i=0}^{\ell-2}w_i + 2p(1-p)w_{\ell-1} + \bigl(\ell(p^2-2p) -(p^2 - 2p +2)+\omega\bigr) w_{\ell} \le 0\ .
\end{equation}

 \smallskip
\noindent
\textbf{The $n_\textrm{d}$-run.}  For the $n_\textrm{d}$-run, choosing any edge will cause the run to decrease in length by 2 with the probability $p^2$,
to decrease in length by 1 with probability $2p(1-p)$, and to remain the same with probability $(1-p)^2$. Thus we obtain
\begin{equation*}
\frac{1}{n}\big( p^2 w_{n -2} + 2p(1-p)w_{n-1} + (1-p)^2w_n \big)n \le (1-\delta)w_n\ .
\end{equation*}

Simplifying this inequality yields
\begin{equation}
\label{e:allminus}
 p^2 w_{n -2} + 2p(1-p)w_{n-1} + (p^2-2p+ \delta)w_n \le 0\ .
\end{equation}

Finally, consider the case where two adjacent runs are separated by a singleton plus.
Suppose the lengths of these runs are $(\ell_1-1)$ and $\ell_2$. If we delete the singleton
plus which separates them, a run of length $\ell_1+\ell_2$ is created. Let us
count this as two runs of length $\ell_1$ and $\ell_2$. In other words, we calculate the
resulting weight as $w_{\ell_1}+ w_{\ell_2}$ whereas the true weight is $w_{\ell_1+\ell_2}$.
We need to know that this underestimates the true cost. This can be done by adding the inequalities
\begin{equation}
\label{e:singleplus}
    w_{\ell_1} + w_{\ell_2} \ge w_{\ell_1 + \ell_2}\ .
\end{equation}

\subsection{Determining the weights}

We now show that we can find appropriate values for the weights $w_\ell$ satisfying
inequalities (\ref{e:singleton}) to  (\ref{e:singleplus}). This will imply that the expectation of
the total weight of the runs in a cycle decreases in expectation after every iteration of the game,
leading to a fast (polynomial) convergence rate.
We also determine a range of $p$ favouring fast convergence.

Solving for $w_{\ell+1}$ in inequality \eqref{e:shortequation} gives the recurrence
\begin{equation}
\label{e:recurrence}
\hat{w}_{\ell+1} = - p(2-p)  \sum_{i=0}^{\ell-2}\hat{w}_i - p(1-p)\hat{w}_{\ell-1} - \frac{1}{2} \bigl(\ell(p^2-2p)
- (p^2 - 2p +2) + \delta n\bigr)  \hat{w}_{\ell}\ ,
\end{equation}
for  $2 \leq \ell \leq n-1$. And, from (\ref{e:singleton}), we have
\begin{equation}
\label{e:recurrencePair}
   \hat{w}_2 = \bigl(1-\tfrac{1}{2}\omega\bigr) \hat{w}_1\ .
\end{equation}

Define $g(\ell)$ by
\begin{equation*}
g(\ell) = \dfrac{\hat{w}_\ell}{\ell}\ .
\end{equation*}
\begin{figure}
\centering
\subfigure[  When $p \ge 0.870$, as $\ell$ increases, $g(\ell)$ decreases initially and then increases exponentially.
 When $p \le 0.869$, $g(\ell)$ decreases continuously. ] {
\label{fig:curveRPS}
\psfrag{ylabel}{\begin{tiny}$g(\ell) \times 10^6$\end{tiny}}
\psfrag{xlabel}{\begin{tiny}$\ell$\end{tiny}}
\includegraphics[width=55mm]{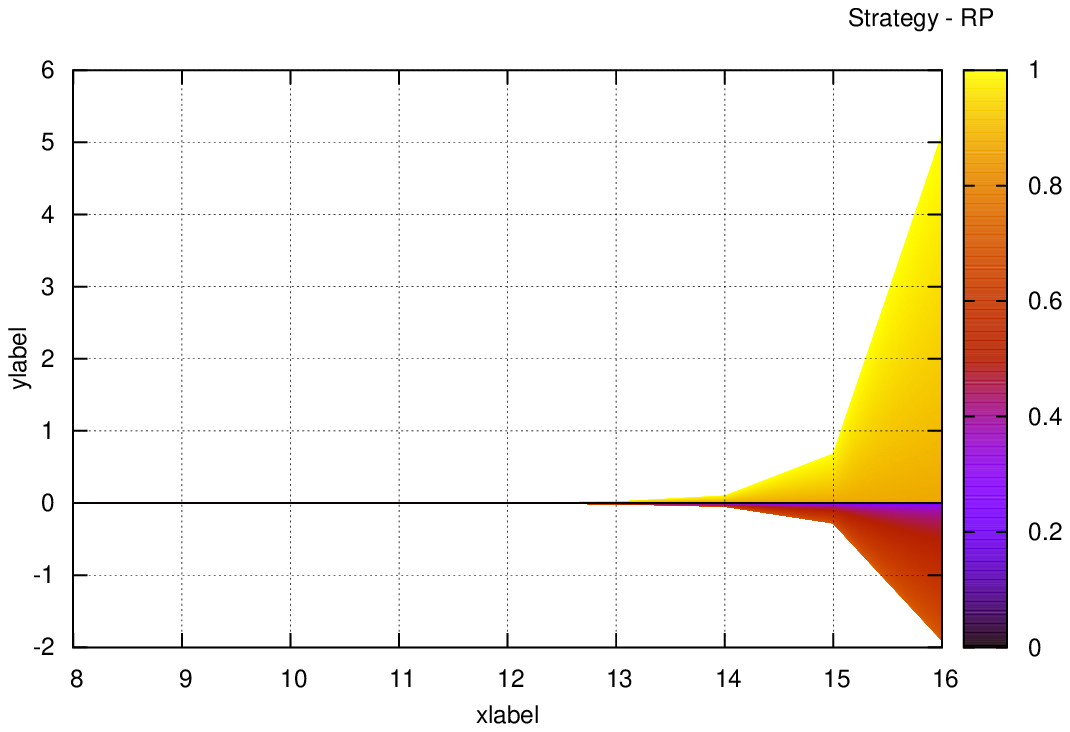}}%
\hspace{1cm}
\subfigure[When $p$ is less than $0.870$, no local minimum was detected,  thus no valid value of $\ell_0$ shown in the figure.
It can be noted that $\ell_0 \le 8$ when $p \ge 0.870$.] {
\label{fig:curveRPSL0}
\psfrag{ylabel}{\begin{tiny}$\ell_0$\end{tiny}}
\psfrag{xlabel}{\begin{tiny}$p$\end{tiny}}
\includegraphics[width=35mm, origin=br,angle=270]{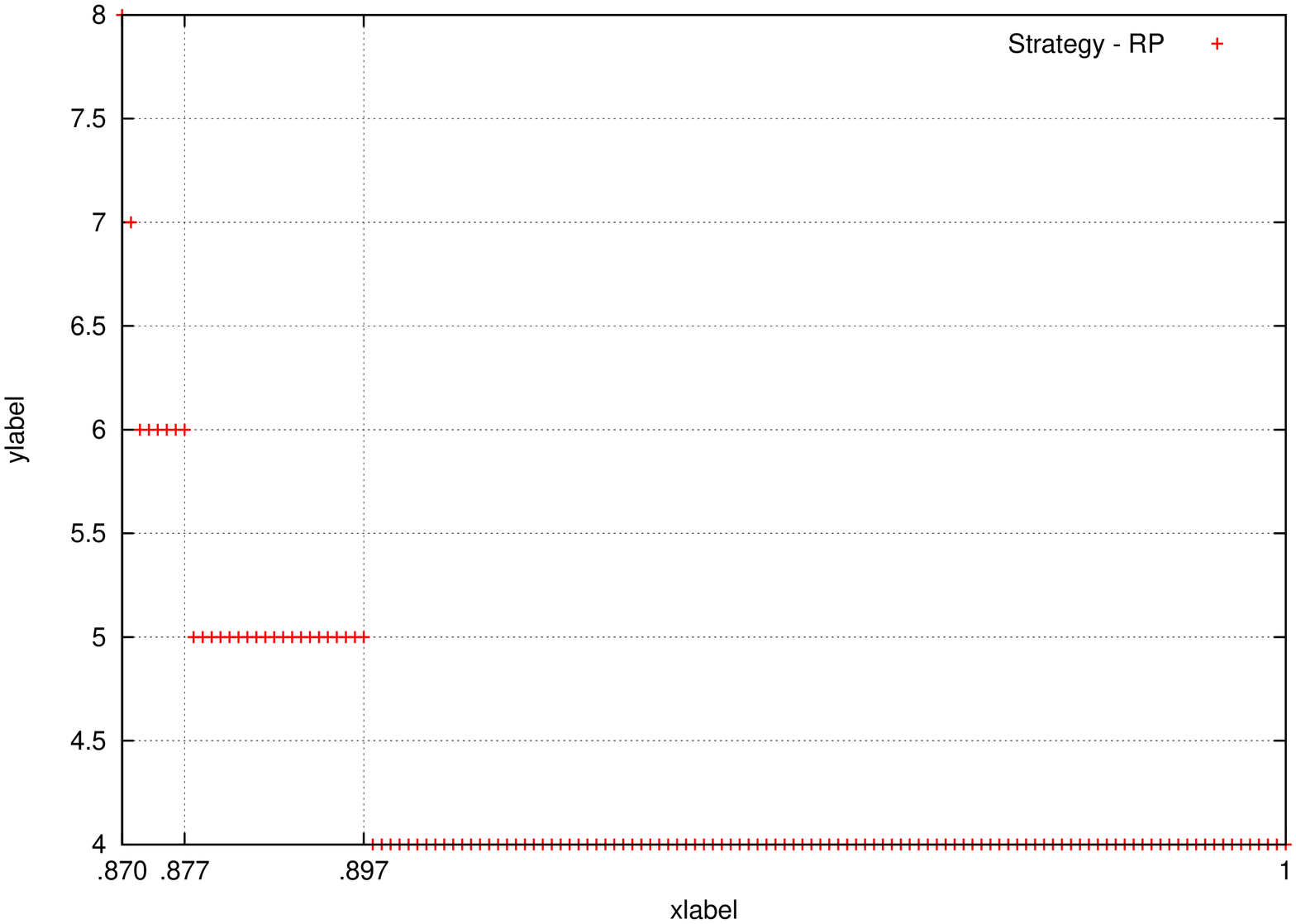}
}
\caption[$g(\ell)$ versus $\ell$ and $\ell_0$ values for RP]{Results from a C++ program: (a) $p$ values plotted as a contour
for $g(\ell)$ versus $\ell$, and (b) experimental values for $\ell_0$.}
\label{fig:curves} 
\end{figure}

A computational study suggests that there exists a $p_0$ such that $g(\ell)$ has a positive minimum
for $p \ge p_0$, and $g(\ell)$ decreases monotonically for $p < p_0$. This is summarised
in Figure~\ref{fig:curveRPS}. Provided this happens, suppose the minimum value $\alpha$
occurs at $\ell = \ell_0(p)$ for $p \ge p_0$.  We will write $\ell_0(p)$ simply as $\ell_0$ for notational
simplicity.  (See Figure~\ref{fig:curveRPSL0}.)
Then

\begin{equation*}
\alpha = g(\ell_0) = \dfrac{\hat{w}_{\ell_0}}{\ell_0}\ .
\end{equation*}

We use this property of the function  $g(\ell)$,i.e.\ having a minimum for high $p$ values, to define
the weights $w_\ell$. Lemma~\ref{lem:minexists} below gives the proof of existence for $\ell_0$. Now,
for $p \ge p_0$, define the weights of the runs as
\begin{equation}
\label{e:weights}
w_\ell = \left\{ \begin{array}{ll}
\hat{w}_\ell & \ \ \textrm{if $\ell \leq \ell_0$,}\\
\alpha \ell &  \ \ \textrm{otherwise,}
\end{array} \right.
\end{equation}
where $\alpha$ is a local minimum of the function $g(\ell)$. The following lemma proves the validity of the
assumption that $g(\ell)$ has a minimum when $p \ge p_0$.

\begin{lemma}
\label{lem:minexists}
There exists $p_0\leq 0.870$ such that $g(\ell)$ has a minimum when $p \ge p_0$.
\end{lemma}

\begin{proof}
To prove the lemma, we determine polynomial functions of $p$ satisfying the first few terms of the recurrences (\ref{e:recurrence}), with seeds $\hat{w}_0 = 0$ and $\hat{w}_1 = 1$.
We use these to find inequalities which determine the range of $p$ such that $g(\ell)$ has a minimum. We then solve these  numerically.
In fact, we use the decrease in $g(\ell)$ at a given $\ell$, which we denote by $h(\ell)$. That is,
\begin{equation*}
h(\ell) = g(\ell+1) - g(\ell)\ .
\end{equation*}

Thus, if $g(\ell)$ has its first local minimum at $\ell = \ell_0$, $h(\ell)$ will be negative
for $\ell = 1,2, \ldots, \ell_0-1$ and positive at $\ell = \ell_0$. For simplicity, we assume
that $\omega = 0$ in the calculations that follow. Now, solving (\ref{e:recurrence}) and (\ref{e:recurrencePair}) for $h(\ell)$, with $\hat{w}_0 = 0$ and $\hat{w}_1 = 1$, we obtain:

\begin{enumerate}
  \item $h(1) = -\tfrac{1}{2}.$ Hence, $h(1)<0$ for all $0\leq p\leq 1$.
  \item $h(2) = -\tfrac{1}{6}+\tfrac{1}{6}p^2.$  Hence, $h(2) \le 0$ for all $0\leq p\leq 1$.
  \item $h(3) = -\tfrac{1}{12} -\tfrac{1}{4}p+{\frac {5}{24}}p^2-\tfrac{1}{8}p^4 +\tfrac{1}{4}p^3$. Hence, $h(3) \le 0$ for all $0\leq p\leq 1$.
  \item $h(4) = -\tfrac{1}{20}-{\frac {7}{20}}p-\tfrac{3}{8}p^2+{\frac {21}{20}}p^3+{\frac {11}
{40}}p^4-\tfrac{3}{5}p^5+{\frac {3}{20}}p^6$. Hence, $h(4) \le 0$  for $0\leq p \leq 0.897.$
\end{enumerate}

Continuing in this way, as $\ell$ goes from 5 to 10, the range of $p$ for which $h(\ell) \le 0$ becomes gradually smaller:
\begin{itemize}
  \item $h(5) \le  0$ for $0\leq p \le 0.877$.
  \item $h(6) \le  0$ for $0\leq p \le 0.871$.
  \item $h(7) \le  0$ for $0\leq p \le 0.870$.
  \item $h(8) \le  0$ for $0\leq p \le 0.869$.
\end{itemize}

Here the upper bounds for $p$  are rounded to three decimal places. Note that $h(8)$ is positive if $p \ge  0.870$.
Therefore, if $p \geq 0.870$, $h(\ell)$ is negative for $1\leq \ell \leq 7$ and positive for $\ell=8$. Thus $g(\ell)$
decreases up to $\ell=8$ and increases at $\ell=9$. Hence, by definition, $\ell_0 \le 8$ when  $p \geq 0.870$,  and the lemma is proved.
\end{proof}

Next we prove two properties of the function $w_\ell$, which will be used later in the proof.

\begin{lemma}
\label{lem:wellIncreasing}
$w_{\ell}$ is a non-decreasing sequence.
\end{lemma}

\begin{proof}
Recall that $w_\ell = \hat{w}_\ell$ for $\ell \le \ell_0$.
 Furthermore, from Lemma~\ref{lem:minexists}, we know that $\ell_0 \le 8$ for $p \ge p_0$. Hence,
we first prove that $w_{\ell}$ is increasing up to $\ell_0=8$, by
proving that $\hat{w}_{\ell}$ is increasing as $\ell$ goes from
$1$ to $8$. This is done by solving the recurrences \eqref{e:recurrence} and (\ref{e:recurrencePair}), with seeds $\hat{w}_0 = 0$ and $\hat{w}_1 = 1$. Here also, for simplicity, we assume
that $\omega = 0$. Let $f(\ell)$ be defined by
\begin{equation*}
f(\ell) = \hat{w}_{\ell+1} - \hat{w}_{\ell} \ .
\end{equation*}

Then it suffices to show that $f(\ell) \ge 0$
for $\ell = 0,1,\ldots,7$. But we have
\begin{enumerate}
  \item $f(0) = 1.$ Hence, $f(0)\ge 0$ for all $0 \le  p \le 1.$
  \item $f(1) = 0.$ Hence, $f(1)\ge 0$ for all $0 \le  p \le 1.$
  \item $f(2) = \tfrac{1}{2}p^2$. Hence, $f(2)\ge 0$ for all $0 \le  p \le 1.$
  \item $f(3) =  -p+p^2+p^3-\tfrac{1}{2}p^4$. Hence, $f(3) \ge 0$  for $ 0.689 \le p \le 1.$
  \item $f(4) =  -2p-\tfrac{3}{2}p^2+\tfrac{11}{2}p^3+\tfrac{5}{4}p^4+\tfrac{3}{4}p^6-3p^5$.
  Hence, $f(4) \ge 0$  for $0.805 \le p \le 1.$
\end{enumerate}

Likewise, we obtain
\begin{itemize}
  \item $f(5) \ge 0$ for $ 0.850 \le p \le 1$.
  \item $f(6) \ge 0$ for $ 0.865 \le p \le 1$.
  \item $f(7) \ge 0$ for $ 0.869 \le p \le 1$.
\end{itemize}

All these together show that $\hat{w}_\ell$ is increasing in the range $\ell = 1,2,\ldots,\ell_0$ for $ 0.869 \le p \le 1$. The lemma then follows from the definition that $w_\ell$ is increasing when $\ell \ge \ell_0$.
\end{proof}

\begin{lemma}
\label{lem:nonincreasing}
$\frac{w_\ell}{\ell}$ is a non increasing sequence.
\end{lemma}

\begin{proof}

From the definition of $\ell_0$, $\frac{\hat{w}_\ell}{\ell}$ is decreasing for $1\leq\ell\leq\ell_0$.
Moreover, $w_\ell = \hat{w}_\ell$ for $\ell \le \ell_0$. Therefore $\frac{w_\ell}{\ell}$ also decreases for $1\leq\ell \le \ell_0$.
When $\ell > \ell_0$, we have $\frac{w_\ell}{\ell} = \alpha$ which is obviously non-increasing.
\end{proof}

The following lemmas show that the inequalities \eqref{e:singleton} to
\eqref{e:singleplus} are satisfied by the proposed weights.

\begin{lemma}
\label{lem:weightincrease}
The weights $w_\ell$ defined in (\ref{e:weights}), with seeds $\hat{w}_0 = 0$ and $\hat{w}_1 = 1$, satisfy  inequalities~(\ref{e:singleton}) and (\ref{e:shortequation}).
\end{lemma}

\begin{proof}

Let $k \geq \ell_0$. Then, from (\ref{e:shortequation}), we get
\begin{equation}
\label{e:weightkplus1}
w_{k+1} \leq - p(2-p)  \sum_{i=0}^{k-2}w_i - p(1-p) w_{k-1}- \dfrac{1}{2} \bigl(k(p^2-2p) - (p^2 - 2p +2) + \omega \bigr) w_k\ .
\end{equation}
Let
\begin{equation}
\label{e:cjsubstitute}
c_i = w_i - i\alpha\ .
\end{equation}

Then, substituting   (\ref{e:weights}) and (\ref{e:cjsubstitute})   into (\ref{e:weightkplus1}), we obtain
\begin{equation*}
\begin{split}
(k+1)\alpha \leq & - p(2-p)  \sum_{i=0}^{k-2}(i\alpha  + c_i) - p(1-p) \bigl((k-1)\alpha +c_{k-1}\bigr) \\
&- \dfrac{1}{2} \bigl(k(p^2-2p) - (p^2 - 2p +2) + \omega \bigr)  k\alpha\ .
\end{split}
\end{equation*}

Simplifying yields
\begin{equation}
\label{e:expandedwkplus1}
k \ge  \dfrac{\alpha + \alpha p + p(2-p) \sum_{i=0}^{k-2} c_i + p(1-p) c_{k-1}}{\alpha \bigl(p- \tfrac{1}{2}\omega\bigr)}\ .
\end{equation}

If $k= \ell_1 (> \ell_0)$, we get

\begin{equation*}
\ell_1 \ge \dfrac{\alpha + \alpha p + p(2-p) \sum_{i=0}^{\ell_1-2} c_i + p(1-p) c_{\ell_1-1}}{\alpha \big(p- \tfrac{1}{2}\omega\big)}\ .
\end{equation*}

But, from (\ref{e:weights}) and (\ref{e:cjsubstitute}), we know that $c_k = 0$ for $k \ge \ell_0$. Since $\ell_1 > \ell_0$, we have $c_{\ell_1-1}=0$ and
$ \sum_{i=0}^{\ell_1-2}c_i = \sum_{i=0}^{\ell_0-1} c_i$. Hence

\begin{equation*}
\ell_1 \ge \dfrac{\alpha + \alpha p + p(2-p) \sum_{i=0}^{\ell_0-1} c_i}{\alpha \big(p- \tfrac{1}{2}\omega\big)} = \ell_1^* \text{ (say). }
\end{equation*}

Thus, inequality (\ref{e:shortequation}) holds for $\ell \ge \ell_1^*$. We also know that (\ref{e:shortequation}) holds when $\ell \leq \ell_0$ by (\ref{e:weights}). Therefore, showing $\ell_1^{*} - \ell_0 \le 1 $ will mean that (\ref{e:shortequation}) is true for all values of $\ell$.
To do this, we determine a lower bound on $\ell_0$ by substituting $k=\ell_0$ into (\ref{e:expandedwkplus1}). Thus
\begin{equation*}
\ell_0 \ge \dfrac{\alpha + \alpha p + p(2-p) \sum_{i=0}^{\ell_0-2} c_i + p(1-p) c_{\ell_0-1}}{\alpha \big(p- \tfrac{1}{2}\omega\big)} = \ell_0^* \text{ (say). }
\end{equation*}

Hence we have

\begin{equation*}
\begin{split}
\ell_1^{*} - \ell_0\ \le\ \ell_1^{*} - \ell_0^{*} &= \dfrac{p(2-p)c_{\ell_0-1} -p(1-p)c_{\ell_0-1}}{\alpha \big(p- \tfrac{1}{2}\omega\big)} \\
& \approx \dfrac{c_{\ell_0-1}}{\alpha} , \ \ \mbox{since $\tfrac{p}{\big(p- \tfrac{1}{2}\omega\big)} \to 1$ when $\omega \to 0$.}
\end{split}
\end{equation*}

Substituting (\ref{e:cjsubstitute}) into the above inequality yields
\begin{equation*}
\begin{split}
\ell_1^{*} - \ell_0\ &\le\ \dfrac{ w_{\ell_0-1}- \alpha(\ell_0-1)}{\alpha} \\
& = 1- \dfrac{w_{\ell_0} - w_{\ell_0-1} }{\alpha}  \\
& \le 1,  \ \  \mbox{since $w_{\ell_0} > w_{\ell_0-1}$ from Lemma~\ref{lem:wellIncreasing}.}
\end{split}
\end{equation*}

Finally, it can easily be verified that (\ref{e:singleton}) holds with $w_1 = \hat{w}_1 = 1$ and $w_2 = \hat{w}_2 = 1-\omega/2$, completing the proof.
\end{proof}

\begin{lemma}
\label{lem:allminus}
The weights $w_\ell$ defined in (\ref{e:weights}), with seeds $\hat{w}_0 = 0$ and $\hat{w}_1 = 1$, satisfy inequality \eqref{e:allminus}.
\end{lemma}
\begin{proof}
Assume $\ell_0\leq 8$ and $n\geq 10$, so $w_\ell=\ell\alpha$ for $\ell=n-2,n-1,n$. Then we require
\begin{equation}
 p^2 (n -2)\alpha  + 2p(1-p)(n-1)\alpha  + \left(p^2-2p+ \dfrac{\omega}{n}\right)n\alpha\ \le\ 0,
\end{equation}
which simplifies to
\begin{equation*}
 \omega\ \le\ 2p,
\end{equation*}
so is satisfied for small enough $\omega$.
\end{proof}

\begin{lemma}
\label{lem:singleplus}
The weights $w_\ell$ defined in (\ref{e:weights}), with seeds $\hat{w}_0 = 0$ and $\hat{w}_1 = 1$, satisfy inequality (\ref{e:singleplus}).
\end{lemma}

\begin{proof}
\begin{equation*}
w_{\ell_1} + w_{\ell_2} \ge \ell_1 \dfrac{w_{\ell_1}}{\ell_1} + \ell_2\dfrac{w_{\ell_2}}{\ell_2} \ .
\end{equation*}

But, from Lemma \ref{lem:nonincreasing},  $\frac{w_\ell}{\ell}$ is a decreasing sequence. Thus,
\begin{equation*}
w_{\ell_1} + w_{\ell_2} \ge \ell_1 \dfrac{w_{\ell_1 + \ell_2}}{\ell_1 + \ell_2} + \ell_2\dfrac{w_{\ell_1 + \ell_2}}{\ell_1 + \ell_2} = w_{\ell_1 + \ell_2},
\end{equation*}
proving the lemma.
\end{proof}

\noindent \textbf{Proof of Theorem~\ref{thm:fastconvRP}: }
Consider an $\ell_\textrm{d}$-run where $ 1 \le \ell \le n$. Denote by $\E[s_\ell]$ the expected weight of the resulting runs. Then, we have
\begin{equation*}
\E[s_\ell]\ \le\ (1-\delta) w_\ell,
\end{equation*}
for $p \ge p_0$, since the weights (\ref{e:weights}) satisfy  the constraints (\ref{e:singleton}) to  (\ref{e:singleplus}) by Lemma~\ref{lem:weightincrease}, Lemma~\ref{lem:allminus} and Lemma~\ref{lem:singleplus}.
Furthermore, Lemma~\ref{lem:minexists} showed that the weights  in RP can be defined this way for $p\geq 0.870$.

Now, suppose the initial state of the cycle is $S_0$. Let $S_t$ be the resulting state of the cycle after $t$ steps. Then total weight after one step of the game is therefore
\begin{equation*}
\begin{split}
\E[\W(S_1)\mid S_0] &= \sum_{\ell} \E[s_\ell] r_\ell \\
&\le  \sum_{\ell}  (1-\delta)w_\ell  r_\ell \\
&= (1-\delta) \W(S_0)\ .
\end{split}
\end{equation*}
Thus, by total expectation, we have
\begin{equation*}
\E[\W(S_1)] \le (1-\delta) \E[\W(S_0)]\ .
\end{equation*}
We have $\W(S) \le n$. To see this, note that for an $\ell_\textrm{d}$-run, $\hat{w}_1 = 1\geq \hat{w}_\ell/\ell=w_\ell/\ell\geq \alpha$ when $1\leq\ell\leq\ell_0$. Thus $\alpha\ell\leq w_\ell\leq\ell$ when $1\leq\ell\leq\ell_0$. In particular, this implies $\alpha\leq 1$. When $\ell\geq \ell_0$, we have $w_\ell=\alpha\ell\leq\ell$, implying $w_\ell\leq\ell$ for all $1\leq\ell\leq n$. Summing this over all runs in $S$, we have $\W(S) \le n$. Since $\delta = \omega/ n$, we have
\begin{equation*}
\E[\W(S_1)] \le \left(1-\dfrac{\omega}{n}\right) \E[\W(S_0)]\ .
\end{equation*}

Applying this for $t$ steps, we obtain
\begin{equation*}
\E[\W(S_t)]\ \le\ \left(1-\dfrac{\omega}{n} \right)^t \E[\W(S_0)]\ \le\ \left(1-\dfrac{\omega}{n} \right)^t n\ \le\ e^{-\tfrac{\omega t}{n}} n\ \leq\ \varepsilon,
\end{equation*}
when
\begin{equation*}
t\ >\  \dfrac{n}{\omega} \log \left( \dfrac{n}{\varepsilon} \right)\ .
\end{equation*}
We also know that, for any $S\neq S^*$, $\W(S)\geq w_1\geq 1$ by Lemma~\ref{lem:wellIncreasing}. Thus, using Markov's inequality, we obtain
\begin{equation*}
\Pr[\W(S_t) \ne 0]\ =\ \Pr[\W(S_t) \ge 1]\ \le\ \E[\W(S_t)]\ \le\ \varepsilon,
\end{equation*}
and the theorem is proved.
\hfill $\square$

\medskip
\begin{remark}
The above requires satisfying (\ref{e:singleton}) to  (\ref{e:singleplus}). These are all linear inequalities.
Therefore, we can solve them by linear programming. Initially, we solved the problem this way, obtaining the same results as above.
\end{remark}
\medskip
\begin{remark}
The problem for SRP can be formulated and solved in the same way as described in
Section~\ref{s:fastconvergence} for RP. We did this and found that the
convergence to cooperation is fast when $p \ge 0.699$. So, the range of $p$
for which the convergence is fast is bigger for SRP than for RP.
 This is somewhat expected because, for a given $p$,
SRP is more forgiving than RP.
\end{remark}

\section{Slow convergence on the cycle}
\label{s:slowconvergence}

In Section~\ref{s:fastconvergence}, we proved that the IPD game converges
to cooperation fast  for high values of $p$. It raises an
interesting research question: how fast or slow
is the convergence when $p$ is small? In this section, we answer this
question by proving Theorem~\ref{thm:slowConRP}, which shows that the convergence to cooperation
takes time exponential in $n$  for small enough $p$. The idea of the proof is to show that it takes exponential time for a plus-run of length $\Omega(n)$ to be formed. (This is done by analysing plus-runs on the cycle. Therefore, in this section, a run refers to a run of ``pluses'' unless otherwise.) It obviously follows that it takes exponential time for the all-cooperate state to be reached.

\subsection{Problem formulation}

Let $\mathcal{R}^i_\ell(t)$ denote the event that a run of $\ell$ pluses (an $\ell_\textrm{c}$-run) starts at position $i$ at time $t$, i.e.\ $S_k = 1$ for $i \le k \le i+\ell-1$ and $S_{i-1}=S_{i+\ell} = -1.$
By the symmetry of the cycle, and the initial configuration, $\Pr(\mathcal{R}^i_\ell)$ will be the
same for all $i$. Let $\delta_j(t)$ denote the event that $S_j$ is a minus at time $t$. i.e.\ $\delta_j(t) = \{  S_j = -1\}$. We will write $\mathcal{R}^i_\ell(t)$ and $\delta_j(t)$ simply as $\mathcal{R}^i_\ell$ and $\delta_j$ respectively, to ease the notation. Then we will define $P^t_\ell$ $(\ell=0,1,\ldots,n-1)$ to be
\begin{equation*}
 P^t_\ell = \Pr ( \mathcal{R}^i_\ell \mid  \delta_{i-1})\qquad(i=0,1,\ldots,n)\ .
\end{equation*}
The conditioning on $\delta_{i-1}$ means that the probability $P^t_\ell$ is an upper bound
on $\Pr ( \mathcal{R}^i_\ell)$. This follows since,
if $S_{i-1} = +1$, a plus-run cannot start at $i$. Recall that $\ell = 0$ means
the length of the plus-run is $0$. Hence,  in particular, $P^t_0$  is an upper bound on the
probability that there are minuses at positions $i-1$ and $i$. An advantage of this approximation is that the $P^t_\ell$ are a probability distribution for
$\ell=1,2,\ldots,n$, whereas the quantities $\Pr(\mathcal{R}^i_\ell)$ do not sum to $1$ in general.

Later in the proof, we will need to calculate an upper bound on the probability that two plus
runs are separated by two minuses. That is, we need to calculate an upper bound on the joint
probability $\Pr( \mathcal{R}^i_{\ell} \wedge  \mathcal{R}^j_m)$ where
$i=j+m+2$. But, we have
\begin{equation}
\label{e:conditional1}
\Pr( \mathcal{R}^i_{\ell} \wedge  \mathcal{R}^j_m) = \Pr( \mathcal{R}^i_{\ell} \mid \mathcal{R}^j_m) \Pr( \mathcal{R}^j_m)\ .
\end{equation}
We will use the fact that, conditional on $\delta_r$,
the $S_q$ for $q > r$ and the $S_k$ for $k < r$ are independent,
if the vertices $k$ and $q$ belong to different plus-runs and
there is at least one more plus-run on the cycle.
Under this condition, changes to the $S_q$ occur independently from
those to the $S_k$, since all steps are independent and affect
only two adjacent vertices. The structure of the cycle means that
changes to the $S_q$ can only be percolated to the $S_k$ through the
vertex $i$, on which we have conditioned. Thus, given $\delta_{i-1}$,
the $\mathcal{R}^i_{\ell}$ is conditionally independent of the $\mathcal{R}^j_m$,
provided there is at least another plus-run on the cycle.
The assumption of having at least three plus-runs holds
initially because the game is started with all-minuses,
which means there are $n$ $0$-runs of pluses. Moreover,  we will then
show that it takes exponential time for a plus-run of length $n/4$ to be formed.
To summarise, we may assume
\begin{equation}
\label{e:conditional2}
\Pr( \mathcal{R}^i_{\ell} \mid \mathcal{R}^j_m) = \Pr( \mathcal{R}^i_{\ell} \mid \delta_{i-1})\ .
\end{equation}
Therefore, from \eqref{e:conditional1} and \eqref{e:conditional2}, we have
\begin{equation*}
\Pr( \mathcal{R}^i_{\ell} \wedge  \mathcal{R}^j_m)\,=\,\Pr( \mathcal{R}^i_{\ell} \mid \delta_{i-1}) \Pr( \mathcal{R}^j_m) \,\leq \,\Pr( \mathcal{R}^i_{\ell} \mid \delta_{i-1}) \Pr( \mathcal{R}^j_m\mid \delta_{j-1} ) \,= \,P^t_\ell P^t_m\ .
\end{equation*}

Note that this inequality and the argument are also applicable when one or both runs are of length $0$ and  separated
by one minus, i.e.\ for $\Pr( \mathcal{R}^i_{\ell} \wedge  \mathcal{R}^{i+\ell+1}_0)$ and
$\Pr( \mathcal{R}^i_0 \wedge  \mathcal{R}^{i+2}_0)$. We use this below
without referring further to the details.

We do not explicitly determine a $p_1$ for which slow convergence occurs. Though this is
possible in principle with our methods, the simpler approach we have chosen already leads
to very cumbersome calculations. Our approach, therefore, is to regard $p$ as small, and use the $O$
and $o$ notation to indicate the order of approximations. Thus there will be some small
enough constant $p_1$ for which our results hold, but we cannot estimate it. In order that
the $O$ etc. notation can be applied to both $p$ and $n$ without confusion, we will assume
that $n>e^{1/p}$.

\medskip

We will first consider short runs. For simplicity, we will leave the investigation of a $0_\textrm{c}$-run to the end of this section, and start with $1_\textrm{c}$-runs.

\medskip

\noindent \textbf{A $1_\textrm{c}$-run.}~Let $R_\textrm{c}$ be a $1_\textrm{c}$-run at position $i$, i.e.\ $R_\textrm{c}=[i,i]$. Choosing either of its outer rim
edges causes $R_\textrm{c}$ to be deleted. On the other hand, $R_\textrm{c}$ is created from a $2_\textrm{c}$-run at position $i-1$ if the edge
$\{i-2, i-1 \}$ is selected and from a $2_\textrm{c}$-run at position $i$ if the edge $\{i+1,i+2\}$ is selected. In addition,
$R_\textrm{c}$ is created from three consecutive minuses at positions $(i-1),i, \text{ and } (i+1)$ with probability $p(1-p)$
if either $\{i-1,i \}$ or $\{i,i+1 \}$ is selected. The probability of finding three consecutive minuses is at most $P_0^t$. Combining all this information, we obtain
\begin{equation}
\label{e:probOfSingletonRPS1}
P^{t+1}_1 = P^t_1 + \dfrac{1}{n}  \bigl(-2 P^t_1 + 2 P^t_2 + 2 p (1-p)P^t_0\bigr)\ .
\end{equation}
Note that the coefficient of $P^t_1$ on the right hand side of (\ref{e:probOfSingletonRPS1}) is positive if
$n \ge 2$ and other two variables $P^t_0$ and $P^t_2$ also have positive coefficients. Hence, using the
upper bounds of these three variables yield an upper bound for $P^{t+1}_1$ as required. Also, as we only need
an upper bound, we have ignored the cases where both $i-2$ and $i-1$ are minuses. In that case, choosing the edge
$\{i-2,i-1\}$ causes the $1_\textrm{c}$-run to increase in length by 2 with probability $p^2$ and by 1 with probability $p(1-p)$, effectively deleting $R_\textrm{c}$.
We will perform similar approximations for the other runs investigated below, without mentioning these details further.

The equation (\ref{e:probOfSingletonRPS1}) is a difference equation with time step $1$. Let us rescale so that the new time step is $1/n$. The difference equation corresponding to the new step size is then as follows.
\begin{equation}
\label{e:probOfSingletonRPS2}
P^{\tfrac{t}{n}+\tfrac{1}{n}}_1 = P^{\tfrac{t}{n}}_1 + \dfrac{1}{n}  \biggl(-2 P^{\tfrac{t}{n}}_1 +
2 P^{\tfrac{t}{n}}_2 + 2p(1-p)P^{\tfrac{t}{n}}_0 \biggr)\ .
\end{equation}

Let  $\tau = \tfrac{t}{n}$ and $h = \tfrac{1}{n}$. Then the equation (\ref{e:probOfSingletonRPS2}) can be written as
\begin{equation}
\label{e:probOfSingletonRPS3}
\dfrac{P^{\tau+h}_1 - P^{\tau}_1}{h} = -2 P^{\tau}_1 + 2P^{\tau}_2 + 2p(1-p) P^{\tau}_0\ .
\end{equation}

Now, the difference equation (\ref{e:probOfSingletonRPS3}) can be approximated
by the following differential equation, with error up to $O(h)=O(1/n)=O(e^{-1/p})$, say, on the right hand side.
\begin{equation}
\label{e:diffEqnSingletonRPS}
 \dfrac{d P_1^\tau}{d \tau} = -2P^{\tau}_1 + 2 P^{\tau}_2+ 2p(1-p)P^{\tau}_0\ .
\end{equation}

\medskip
\noindent \textbf{A $2_\textrm{c}$-run.}~Let $R_\textrm{c}$ be a $2_\textrm{c}$-run starting at position $i$, hence $R_\textrm{c}=[i,i+1]$. Similarly to a
$1_\textrm{c}$-run, choosing either of the two outer rim edges of $R_\textrm{c}$ causes the run to decrease in length by 1, reducing the
number of $2_\textrm{c}$-runs on the cycle by 1. $R_\textrm{c}$ is created by
choosing the outer rim edge $\{i-2,i-1\}$ of a $3_c$-run at $(i-1)$. Similarly, $R_\textrm{c}$ is created by
choosing the outer rim edge $\{i+2,i+3\}$ of a $3_c$-run at $i$.  In addition, $R_\textrm{c}$ can be created from a singleton
plus adjacent to a pair of minuses. This happens if the edge connecting the pair of minuses is selected and only the
minus next to the  singleton plus becomes a plus. The probability for this, given that the corresponding edge has
been selected, is $p(1-p)$. Finally, $R_\textrm{c}$ is created with probability $p^2$ by selecting the middle edge of four
consecutive minuses. The probability of having four consecutive minuses at location $i$ is at most ${P^t_0}^2$. Therefore we get
\begin{equation}
\label{e:probOfPairRPS1}
P^{t+1}_2 = P^t_2 + \dfrac{1}{n}  \Bigl(-2P^t_2 +
2P^t_3 + 2p(1-p)P^t_1P^t_0 + p^2  {P^t_0}^2 \Bigr)\ .
\end{equation}
As before, rescaling and approximating, we obtain
\begin{equation}
\label{e:diffEqnPairRPS}
\dfrac{d P_2^\tau}{d \tau}  = -2P^{\tau}_2 +
2P^{\tau}_3 + 2p(1-p)P^{\tau}_1P^{\tau}_0 + p^2 {P^{\tau}_0}^2\ .
\end{equation}

\medskip
\noindent \textbf{An $\ell_\textrm{c}$-run, where $ \ell \geq 3.$}~
Suppose $R_c = [i,j]$ is an $\ell_\textrm{c}$-run for some $\ell \ge 3$. Selecting either of the two outer rim edges
causes  $R_c$ to decrease in length by $1$. On the other hand, an $(\ell+1)_c$-run starting at position $(i-1)$
is turned into an $\ell_\textrm{c}$-run starting at position $i$ if the edge $\{i-2,i-1\}$ is chosen; and, an $(\ell+1)_c$-run
starting at position $i$ becomes an $\ell_\textrm{c}$-run starting at the same position $i$ if the edge $\{j+1, j+2\}$ is chosen.
Also, if there is a $0_\textrm{c}$-run at $i$ and an $(\ell-1)_c$-run at $i+1$, choosing the edge $\{i-1,i\}$ will create an
$\ell_\textrm{c}$-run starting at location $i$ with probability $p(1-p)$. We will get the same result if these two runs are
in the reverse order: $(\ell-1)_c$-run at $i$ and a $0_\textrm{c}$-run at $j$. If there is an $(\ell-2)_c$-run starting at $(i+2)$
and there are minuses at $i-1$, $i$, and $i+1$,  then choosing the edge $\{i, i+1\}$ produces an $\ell$-run at $i$  with
probability $p^2$. Similarly if there is an $(\ell-2)_c$-run at position $i$ and there are minuses at positions $j-1$, $j$
and $j+1$, the run increases in length by $2$ with probability $p^2$ if the edge $\{j-1,j\}$ is selected. Finally a $k_c$-run
and an $(\ell-2-k)_c$-run, $1 \le k \le \ell-3$, at positions $i$ and $(i+k+2)$ respectively merge with probability $p^2$, introducing an $\ell_\textrm{c}$-run, if
the edge between the runs, namely $\{i+k,i+k+1\}$, is selected.  Thus we have
\begin{equation*}
P^{t+1}_{\ell} = P^t_{\ell} + \dfrac{1}{n}  \biggl(-2 P^t_{\ell} +
2P^t_{\ell+1} + 2p(1-p) P^t_{\ell-1} P^t_0 + 2p^2  P^t_{\ell-2} P^t_0 +p^2
\sum_{k=1}^{\ell-3}
P^t_k P^t_{\ell-2-k}\biggr)\ .
\end{equation*}

This could be written as
\begin{equation}
\label{e:probofEllRunRPS1}
P^{t+1}_{\ell}
= P^t_{\ell} + \dfrac{1}{n}  \biggl(-2P^t_{\ell} +
2P^t_{\ell+1} + 2p(1-p) P^t_{\ell-1} P^t_0 +p^2
\sum_{k=0}^{\ell-2}
 P^t_k  P^t_{\ell-2-k}\biggr)\ .
\end{equation}
Here also, we have used the fact that the probability of finding three consecutive minuses is at most $P_0^t$. Observe that, in this form, the difference equation (\ref{e:probOfPairRPS1})
 is equivalent to (\ref{e:probofEllRunRPS1}) when $\ell = 2$. Therefore, we can use (\ref{e:probofEllRunRPS1}) for $\ell = 2$ also.

\medskip
 \noindent \textbf{A $0_\textrm{c}$-run.}~Finally, consider a run of length zero, i.e.\ a $0_\textrm{c}$-run. Recall that we have defined $P^t_0$  and $P^t_1$ to be upper bounds on the probability of finding a $0_\textrm{c}$-run and $1_\textrm{c}$-run respectively, at position $i$ at time $t$.
 Now let $\bar{P}^t_0$ and $\bar{P}^t_1$ denote the exact values of these probabilities respectively, i.e.\ $\bar{P}^t_0 = \Pr(\mathcal{R}^i_0) $  and $\bar{P}^t_1 = \Pr(\mathcal{R}^i_1)$. We can now examine the dynamics of a $0_\textrm{c}$-run. A $0_\textrm{c}$-run at position $i$ means there are minuses at
 positions $(i-1)$ and $i$. Then, $i+1$ can be a minus or a plus. It is not difficult to verify that, if it is a minus then the
 $0_\textrm{c}$-run might be deleted with probability $(3p-p^2)/n$, and if it is a plus , the $0_\textrm{c}$-run might be deleted with probability
 $( 2p-p^2)/n $. Also note that probability of finding each of these configurations is at most $\bar{P}^t_0$. On the creation
 side, a $0_\textrm{c}$-run at $i$ could be created from a $1_\textrm{c}$-run at $i$ with probability $2/n$ and from any longer plus-runs at position
 $i$ with probability $1/n$. By definition, the probability of finding a $1_\textrm{c}$-run at $i$ is $\bar{P}^t_1$.  It then follows that
 the probability of finding a plus-run of length greater than $2$ at position $i$ is $(1-\bar{P}^t_0-\bar{P}^t_1)$ . Hence we obtain

\begin{equation*}
P^{t+1}_0 =  \bar{P}^t_0 + \dfrac{1}{n}  \bigl(-\bar{P}^t_0 (3p-p^2)-\bar{P}^t_0 ( 2p-p^2)
+ 2\bar{P}_1^t +(1-\bar{P}^t_0-\bar{P}^t_1)  \bigr)\ .
\end{equation*}

Thus
\begin{equation}
\label{e:probOfZeroRPS1}
P^{t+1}_0 =  \bar{P}^t_0 \biggl( 1 - \dfrac{1+5p-2p^2}{n}  \biggr) + \dfrac{1}{n} (1 + \bar{P}_1^t).
\end{equation}
Note that $P^{t+1}_0$ in (\ref{e:probOfZeroRPS1}) is an upper bound. Furthermore, the coefficient of
$\bar{P}^t_0$ is positive when $n \ge 1+5p-2p^2$, and the coefficient of $\bar{P}_1^t$
is also positive. We can therefore replace $\bar{P}^t_0$ and $\bar{P}^t_1$ with their upper bounds $P^t_0$
and $P^t_1$ respectively, obtaining
\begin{equation}
\label{e:probOfZeroRPS2}
P^{t+1}_0 =  P^t_0 + \dfrac{1}{n}  \bigl(-(1+5p-2p^2)P^t_0 + P_1^t + 1 \bigr)\ .
\end{equation}
Hence we get
\begin{equation}
\label{e:diffEqnZeroRunRPS}
\dfrac{d P_0^\tau}{d \tau} = -(1+5p-2p^2)P^{\tau}_0 + P_1^{\tau} + 1\ .
\end{equation}

\subsection{The analysis}
In the previous section we modelled the game dynamics by a set of differential equations.
We first solve the ones corresponding to the runs of length shorter than $3$.
\begin{lemma}
\label{lem:specificSolutionRPS}
If the game is started in the all-minuses configuration, the solution to the system of
differential equations \eqref{e:diffEqnSingletonRPS},  \eqref{e:diffEqnPairRPS}
and \eqref{e:diffEqnZeroRunRPS} is given by
\scriptsize
\begin{equation*}
\begin{bmatrix}
P_0^\tau\\[1ex]
P_1^\tau\\[1ex]
P_2^\tau
\end{bmatrix}
=
\begin{bmatrix}
1 + (-4+3 e^{-\tau} + e^{-2\tau} ) p+ \left(\tfrac {37}{2}+ ( 5 e^{-2\tau}-9 {e^{-\tau}}
) \tau + 2 e^{-2\tau}{\tau}^2-31 e^{-\tau}+\tfrac {25}{2}e^{-2\tau}
\right) p^2+o( p^2)\ \\[0.5ex]
( 1-e^{-2\tau} ) p + \left(-\tfrac {7}{2}+6e^{-\tau}-\tfrac {5}{2}e^{-2\tau}-e^{-2\tau}\tau
-2e^{-2\tau}{\tau}^2 \right) p^2+o ( p^2)\ \\[0.5ex]
\left( \tfrac {3}{2}-\tfrac {3}{2}e^{-2\tau}-2 e^{-2\tau} \tau \right) p^2+o(p^2)
\end{bmatrix}\ .
\end{equation*}

\end{lemma}

\begin{proof}
Note that the differential equations (\ref{e:diffEqnSingletonRPS}) and (\ref{e:diffEqnZeroRunRPS})   are linear, while
(\ref{e:diffEqnPairRPS}) is nonlinear. Fortunately, we can approximately linearise  \eqref{e:diffEqnPairRPS} using some knowledge of the system.

We approximate the solutions with error terms $o(p^2)$. Then, assuming $P_0^\tau = 1+O(p)$, $P_1^\tau = O(p)$ and $P_3^\tau = o(p^2)$
linearises (\ref{e:diffEqnPairRPS}). The linearised version is given by
\begin{equation*}
\dfrac{d P_2^\tau}{d \tau}  = 2p(1-p)P^{\tau}_1 -2 P^{\tau}_2 + p^2 +o(p^2)\ .
\end{equation*}

Hence, for short runs, we have the following nonhomogeneous linear system of first order differential equations.
\begin{align}
\nonumber
\dfrac{d P_0^\tau}{d \tau} &= -(1+5p-2p^2)P^{\tau}_0 + P_1^{\tau} + 1 \ . \\
\nonumber
\dfrac{d P_1^\tau}{d \tau} &= 2p(1-p)P^{\tau}_0-2P^{\tau}_1 + 2P^{\tau}_2 \ .\\
\nonumber
\dfrac{d P_2^\tau}{d \tau} &= 2p(1-p)P^{\tau}_1-2P^{\tau}_2 + p^2 +o(p^2) \ .
\end{align}

In matrix form, the system can be written as

\begin{equation*}
\dfrac{d }{d \tau}
\begin{bmatrix}
P_0^\tau\\
P_1^\tau\\
P_2^\tau
\end{bmatrix}
=
\begin{bmatrix}
-(1+5p-2p^2) & 1 & 0\\
2p(1-p) & -2 & 2\\
0 & 2p(1-p) & -2
\end{bmatrix}
\begin{bmatrix}
P_0^\tau\\
P_1^\tau\\
P_2^\tau
\end{bmatrix}
+
\begin{bmatrix}
1\\
0\\
p^2+o(p^2)
\end{bmatrix}\ .
\end{equation*}

Let us denote this system by
\begin{equation}
\label{e:coupledDiff}
{\Pt}' = \At \Pt+\Ft.
\end{equation}
Since the game is started  with the all-minuses configuration, we have the initial condition
\begin{equation*}
\Pt(0) = \begin{bmatrix} 1\\0\\0\end{bmatrix}\ .
\end{equation*}
Thus we have an initial value problem which we solve by the method of decoupling.
We first find the eigenvalues and eigenvectors of $ \At$. The characteristic polynomial of $ \At$
is
\begin{equation*}
{\lambda}^3+ (5+5p-2p^2) {\lambda}^2+( 8+14p-2p^2) \lambda +4 +12p-20p^2+28p^3-8p^4\ .
\end{equation*}

An analysis of this cubic polynomial shows that all three roots are real,
different, and negative for small $p$. The eigenvalues of $\At$ are
\begin{equation*}
\begin{bmatrix}
\lambda_1\\
\lambda_2\\
\lambda_3
\end{bmatrix}
=
\begin{bmatrix} -1-3p+14p^2+o(p^2) \ \\[0.5ex]
 -2-2\sqrt{p}-p+\tfrac{11}{4}p^{3/2}-6p^2+o (p^2)\ \\[0.5ex]
-2+2\sqrt{p}-p-\tfrac{11}{4}p^{3/2}-6p^2+o(p^2)
\end{bmatrix}\ .
 \end{equation*}

Now, the eigenvector of $\At$ corresponding to eigenvalue $\lambda_1$ is
\begin{align*}
\mathbf{e}_1\ &=\ \begin{bmatrix}\ \tfrac{1}{4}p^{-2}-2p^{-1}+6+O(p)\ \\
\tfrac{1}{2}p^{-1}-1+6p+O(p^2)\\
1
\end{bmatrix}\\
&=\ \frac{1}{4p^2}\begin {bmatrix}\ 1-8p+24p^2 +O(p^3)\ \\
2p-4p^2 +O(p^3)\\
4p^2
\end{bmatrix}.
\end{align*}
The eigenvector corresponding to eigenvalue $\lambda_2$ is
\begin{align*}
\mathbf{e}_2\ &=\ \begin{bmatrix}\ p^{-1/2}-\frac{3}{2}+\tfrac{53}{8}p^{1/2}-13p+\tfrac{4167}{128}p^{3/2}
-\tfrac{101}{2}p^2+\tfrac{40525}{1024}p^{5/2}+O(p^3)\ \\[0.5ex]
 -p^{-1/2}-\frac{1}{2}+\tfrac{3}{8}p^{1/2}-\tfrac{7}{2}p+
\tfrac{1001}{128}p^{3/2}-\tfrac{43}{2}p^2+\tfrac{45627}{1024}p^{5/2}+O(p^3)\ \\[0.5ex]
1
\end{bmatrix}\\[1ex]
&=\ \frac{1}{\sqrt{p}}\begin{bmatrix}\ 1-\frac{3}{2}\sqrt{p}+\tfrac{53}{8}p-13p^{3/2}
+\tfrac{4167}{128}p^2-\tfrac{101}{2}p^{5/2}+O(p^3)\ \\[0.5ex]
-1-\frac{1}{2}\sqrt{p}+\tfrac{3}{8}p-\frac{7}{2}p^{3/2}+\tfrac{1001}{128}p^2
-\tfrac{43}{2}p^{5/2} +O(p^3)\\[0.5ex]
\sqrt{p}
\end{bmatrix}.
\end{align*}
Finally, the eigenvector corresponding to eigenvalue $\lambda_3$ is

\begin{align*}
\mathbf{e}_3\ &=\ \begin{bmatrix} \ -{\tfrac {1}{\sqrt {p}}}-\tfrac{3}{2}-\tfrac {53}{8}
\sqrt {p}-13 p-{\tfrac {4167}{128}}p^{3/2}-\tfrac {101}{2}p^2-\tfrac {40525}{1024}p^{5/2}+O ( p^3)\ \\[0.5ex]
\ \tfrac {1}{\sqrt {p}}-\frac{1}{2}-\tfrac{3}{8}\sqrt {p}-\frac{7}{2}p-\tfrac {1001}{128}p^{3/2}-\tfrac {43}{2}p^2-\tfrac {45627}{1024}{p}^{5/2}+O ( p^3) \ \\[0.5ex]
1
\end{bmatrix}\\[1ex]
&=\ \frac{1}{\sqrt{p}}\begin{bmatrix}\
-1-\frac{3}{2}\sqrt {p}-{\tfrac {53}{8}}p-13p^{3/2}-\tfrac {4167}{128}p^2-\tfrac {101}{2}p^{5/2} +O(p^3) \ \\[0.5ex]
1-\frac{1}{2}\sqrt {p}-\tfrac{3}{8}p-\frac{7}{2}{p}^{3/2}-{\tfrac {
1001}{128}}p^2-{\tfrac {43}{2}}p^{5/2}+ O ( p^3)\ \\[0.5ex]
\sqrt {p}
\end{bmatrix}.
\end{align*}

We now form the matrix $\Tt$ whose columns are constant multiples of the eigenvectors of $\At$. That is
\begin{equation*}
\Tt = \quad [ \quad 4p^2\mathbf{e}_1 \quad\quad \sqrt {p}\mathbf{e}_2 \quad\quad \sqrt {p}\mathbf{e}_3\quad]\ .
\end{equation*}

Since all three eigenvalues are different, the eigenvectors $\mathbf{e}_1, \mathbf{e}_2, \text{ and } \mathbf{e}_3$  are linearly
independent. Hence the matrix $\Tt $ is non-singular and $\Tt^{-1}$ exists. Let us calculate the determinant
of $\Tt$ to confirm that the approximated $\Tt$ is non-singular.
\begin{equation*}
\det \Tt = -2\sqrt {p}+\tfrac {51}{4}p^{3/2}-\tfrac {4663}{64}p^{5/2} \ne 0\ .
\end{equation*}
Now, we calculate the inverse of the matrix $\Tt$. Since the determinant of $\Tt$ is $O(p^{1/2})$ and
 $\Tt$ is accurate up to $O(p^{5/2})$, $\Tt^{-1}$ will be correct up to $O(p^2)$.

\scriptsize
\begin{equation*}
\Tt^{-1} = \begin{bmatrix}\ 1+6p-6p^2+o(p^2)&1+13p+79p^2+o(p^2)&2+32
p+226p^2+o(p^2)\ \\[0.5ex]
p-2p^{3/2}+\tfrac {43}{8}p^2+o(p^2)&-\tfrac{1}{2}+\tfrac {13}{16}p-2p^{3/2}+\tfrac {2149
}{256}p^2 +o(p^2)&\tfrac {1}{2\sqrt {p}}-\tfrac{1}{4}+\tfrac {5}{32}p-4p^{3/2}+\tfrac {8885}{512}p^2 +o(p^2)\ \\[0.5ex]
-p-2p^{3/2}-\tfrac {43}{8}p^2+o(p^2)&\tfrac{1}{2}-\tfrac {13}{16}p-2p^{3/2}-\tfrac {2149}{256}p^2+o(p^2)&\tfrac {1}{2\sqrt {p}}+\tfrac{1}{4}-\tfrac {5}{32}p-4p^{3/2}-\tfrac {8885}{512}p^2+o(p^2)\ \\[0.5ex]
\end{bmatrix}\ .
\end{equation*}

\normalsize
Ignoring the error terms, we can verify that $\Tt^{-1}\At\Tt$ is a diagonal matrix whose diagonal
elements are the eigenvalues of $\At$, concurring with the theory. That is
\begin{equation*}
\Tt^{-1}\At\Tt = \begin{bmatrix}\
\lambda_1&0&0\ \\
0&\lambda_2&0\ \\
0&0&\lambda_3
\end{bmatrix}\ .
\end{equation*}

 Let $\Pt = \Tt \Yt$.   Then we have a new system of differential equations given by
\begin{equation}
\label{e:decoupledDiff}
  \Yt^{'} = \Dt\Yt + \Gt,
\end{equation}
with initial condition $\Yt(0) = \Tt^{-1}\Pt(0)$ where $\Dt = \Tt^{-1}\At\Tt$ and $\Gt = \Tt^{-1}\Ft$. Hence

\begin{equation*}
\Yt(0) = \Tt^{-1} \Pt(0)
=
\begin{bmatrix}\ 1+6p-6p^2+o(p^2) \\
p-2p^{3/2}+\tfrac {43}{8}p^2+o(p^2)\ \\[0.5ex]
-p-2p^{3/2}-\tfrac{43}{8}p^2+o(p^2)
\end{bmatrix}\ .
\end{equation*}

 We also know that

\begin{equation*}
\Ft =  \begin{bmatrix}\
1\ \\
0\ \\
p^2 +o(p^2)
\end{bmatrix}\ .
\end{equation*}

Thus

\begin{equation*}
\Gt = \Tt^{-1} \Ft = \begin{bmatrix}\ 1+6p-4p^2+o(p^2)\ \\
p-\tfrac{3}{2}p^{3/2}+\tfrac {41}{8}p^2+o(p^2) \ \\[0.5ex]
-p -\tfrac{3}{2}p^{3/2}-\tfrac {41}{8}p^2+o(p^2)
\end{bmatrix}.
 \end{equation*}

Now, solving the three decoupled differential equations \eqref{e:decoupledDiff} yields

\scriptsize
\begin{equation*}
\Yt =  \begin{bmatrix}\  1+ ( 3+3e^{-\tau} ) p+ ( 1
-7e^{-\tau}-9e^{-\tau}\tau ) p^2+o(p^2)\ \\[0.5ex]
\left( \tfrac{1}{2}+ \tfrac{1}{2}e^{-2\tau} \right) p+ \left( -e^{-2\tau}\tau- \tfrac{3}{4}e^{-2\tau}-
 \tfrac{5}{4} \right) p^{3/2}+ \left( e^{-2\tau}\tau+e^{-2\tau}{\tau}^2+\tfrac {
29}{16}e^{-2\tau}+\tfrac {57}{16} \right) p^2+o(p^2)
\ \\[0.5ex]
\left( - \tfrac{1}{2}- \tfrac{1}{2}e^{-2\tau} \right) p+ \left( -e^{-2\tau}\tau- \tfrac{3}{4}e^{-2\tau}- \tfrac{5}{4} \right) p^{3/2}+ \left( -e^{-2\tau}
\tau-{e^{-2\tau}}{\tau}^2-\tfrac {29}{16}e^{-2\tau}-\tfrac {57}{16}
 \right) p^2 +o(p^2)\end{bmatrix}\ .
 \end{equation*}

\normalsize
Finally, the solution for \eqref{e:coupledDiff} can be computed using $\Pt = \Tt\Yt$.  What is remaining to be shown
is that the three assumptions used in the proof are valid. They are: $P_3^\tau = o(p^2),\ P_1^\tau = O(p)$
and $P_0^\tau = 1-O(p)$. The assumption on $P_3^\tau$ is validated in Lemma~\ref{lem:upperBoundsRPS}.
Let us consider the other two here. The final solution confirms that our assumptions are valid at any time
$\tau$ if they were valid initially. Clearly the assumptions hold initially as, at time $\tau=0$, we have
$P_0^0 = 1$ and  $P_1^0 = 0$. Hence, the final solution holds for any $\tau$ and the proof is complete.
\end{proof}

We will use generating functions to solve the recurrence (\ref{e:probofEllRunRPS1}).
Let the function $F(x,t)$ be defined by
\begin{equation*}
F(x,t) = \sum_{\ell = 0}^{\infty} P_\ell^t x^{\ell}\ .
\end{equation*}

Now, multiplying (\ref{e:probofEllRunRPS1}) by $x^{\ell+1}$ and summing over all $\ell \ge 2$,
we obtain
\begin{equation}
\label{e:recurrence1}
\begin{split}
\sum_{\ell = 2}^{\infty} P^{t+1}_{\ell}x^{\ell+1} = &\sum_{\ell = 2}^{\infty}P^t_{\ell}x^{\ell+1}
+ \dfrac{1}{n}  \biggl( -2 \sum_{\ell = 2}^{\infty} P^t_{\ell}x^{\ell+1} +
2\sum_{\ell = 2}^{\infty} P^t_{\ell+1}x^{\ell+1} \\ &+ 2p(1-p)\sum_{\ell = 2}^{\infty}P^t_{\ell-1}
P^t_0 x^{\ell+1} +p^2\sum_{\ell = 2}^{\infty}
x^{\ell+1} \sum_{k=0}^{\ell-2}P^t_k P^t_{\ell-2-k}\biggr)\ .
\end{split}
\end{equation}
The indices of (\ref{e:recurrence1}) can be adjusted to get
\begin{equation}
\label{e:recurrence2}
\begin{split}
x\sum_{i = 2}^{\infty} P^{t+1}_ix^i &= x\sum_{i = 2}^{\infty}P^t_ix^i + \dfrac{1}{n}  \biggl(-2x\sum_{i = 2}^{\infty} P^t_ix^i +
2\sum_{i = 3}^{\infty} P^t_ix^i\\ &+ 2p(1-p)x^2 P^t_0 \sum_{i = 1}^{\infty} P^t_i x^i +p^2\sum_{i=0}^{\infty}x^{i+3} \sum_{k=0}^iP^t_k P^t_{i-k}\biggr)\ .
\end{split}
\end{equation}
Note that the last term in (\ref{e:recurrence2}) can be thought of as relating the sequence $P_\ell^t$ to its own convolution, thus can be replaced by their product, obtaining
\begin{equation*}
\begin{split}
x\sum_{i = 2}^{\infty} P^{t+1}_ix^i = &x\sum_{i = 2}^{\infty}P^t_ix^i  + \dfrac{1}{n}  \biggl(-2x \sum_{i = 2}^{\infty} P^t_ix^i +
2\sum_{i = 3}^{\infty} P^t_ix^i \\ &+  2p(1-p)x^2P^t_0 \sum_{i = 1}^{\infty} P^t_i  x^i +p^2x^3\biggl(\sum_{n=0}^{\infty} P^t_nx^n\biggr)\biggl(\sum_{n=0}^{\infty} P^t_nx^n\biggr)\biggr)\ .
\end{split}
\end{equation*}
Hence
\begin{equation*}
\begin{split}
x&\bigl(F(x,t+1)-P_0^{t+1}-P_1^{t+1}x\bigr) = x\bigl(F(x,t)-P_0^t-P_1^t x\bigr) + \dfrac{1}{n} \bigl(-2x\bigl(F(x,t)-P_0^t-P_1^t x\bigl)  \\ &+
2\bigl(F(x,t)-P_0^t-P_1^tx-P_2^tx^2\bigr)  + 2p(1-p)x^2 P^t_0 \bigl( F(x,t) - P^t_0\bigl) +p^2x^3F(x,t)^2 \bigr)\ .
\end{split}
\end{equation*}
This can be rearranged to get
\begin{equation*}
\begin{split}
x&\dfrac{F(x,t+1) -  F(x,t)}{\tfrac{1}{n}} -  x\dfrac{(P_0^{t+1} - P_0^t)}{\tfrac{1}{n}} - x^2 \dfrac{(P_1^{t+1} - P_1^t)}{\tfrac{1}{n}} =  -2x\bigl(F(x,t)-P_0^t-P_1^tx\bigl)  \\ &+
2\bigl(F(x,t)- P_0^t - P_1^tx-P_2^tx^2\bigl)  + 2p(1-p)x^2 P^t_0\bigl( F(x,t) - P^t_0\bigr)\bigr) + p^2x^3F(x,t)^2\ .
\end{split}
\end{equation*}

Substituting (\ref{e:probOfSingletonRPS1}) and (\ref{e:probOfZeroRPS2}) into the above equation yields
\begin{equation}
\label{e:diffEqnFinalRPS1}
\begin{split}
x &\dfrac{F(x,t+1)-  F(x,t)}{\tfrac{1}{n}} = p^2x^3F(x,t)^2 + 2F(x,t)\bigl(1-x+x^2P_0^tp(1-p)\bigl)   \\
&-2x^2p{P_0^t}^2(1-p)+P_0\bigl(x(1-5p+2p^2)+2x^2p(1-p) -2\bigl) -P_1^tx + x\ .
\end{split}
\end{equation}

Now, let $y(\tau)= F(x, t)$ where $\tau = \tfrac{t}{n}$ as defined before. Then, approximating and rescaling, we get
\begin{equation}
\label{e:diffEqnFinalRPS}
\begin{split}
x\dfrac{d y(\tau)}{d\tau} = p^2&x^3y(\tau)^2 + 2y(\tau)\bigl(1-x+x^2P_0^\tau p(1-p)\bigr)-2x^2{P_0^\tau}^2 p(1-p)\\
&+P_0^\tau\bigl(x(1-5p+2p^2)+2x^2p(1-p) -2\bigr)  -P_1^\tau x + x\ .
\end{split}
\end{equation}

The following lemma proves that $y(\tau)$ has a radius of convergence greater than 1.

\normalsize
\begin{lemma}
\label{lem:generatingfunctionRPS}
The generating function $y(\tau)$ is bounded above and converges for
some $x > 1$.
\end{lemma}

\begin{proof}

Without loss of generality, let $x= 1+p^3$. Substituting this value into the differential equation
(\ref{e:diffEqnFinalRPS}) gives
\begin{equation}
\label{e:diffEqnYRPS1}
\dfrac{d y(\tau)}{d\tau} = p^2 y(\tau)^2 + 2y(\tau)P_0^\tau p(1-p) -2{P_0^\tau}^2p(1-p)-P_0^\tau(1+3p) -P_1^\tau + 1 + o(p^2)\ .
\end{equation}

Differential Equation \eqref{e:diffEqnYRPS1} is nonlinear. But, we can linearise this by assuming $y(\tau) = 1+O(p)$.
This assumption will be validated later. Under this assumption, the nonlinear term
\begin{equation*}
p^2 y(\tau)^2 = p^2 + o(p^2)\ .
\end{equation*}

Then, substituting the solutions for $P_0^\tau$ and $P_1^\tau$ from Lemma~\ref{lem:specificSolutionRPS} into  \eqref{e:diffEqnYRPS1} and simplifying, we get the following linear differential equation.
\begin{equation}
\label{e:diffEqnYRPS2}
\begin{split}
\dfrac{d y(\tau)}{d\tau} &-  \bigl( 2p - (10-2e^{-2\tau}-6e^{-\tau}) p^2\bigr) y ( \tau ) = ( -2-3{e^{-\tau}} ) p\ \\ & + ( 16-17 e^{-2\tau} + 4e^{-\tau}
  + 9e^{-\tau}\tau -4e^{-2\tau}\tau ) p^2 + o(p^2)\ .
 \end{split}
\end{equation}

This is a first order linear differential equation which could be solved by the method of integrating factor.
The integrating factor is
\begin{equation*}
\mu(\tau) = e^{\int \!-2p+ ( 10-2e^{-2\tau}-6e^{-\tau}
 ) p^2{d \tau}}
= e^{-2\tau p+ ( 10\tau+e^{-2\tau}+6e^{-\tau} ) p^2}\ .
\end{equation*}

Using Taylor approximations, we can approximate the integrating factor and its inverse to get

\begin{equation}
\label{e:integFactorApprox}
\mu(\tau)  = 1-2\tau p+ ( 10\tau+2{\tau}^2+6e^{-\tau} +e^{-2\tau}) p^
2+ o ( p^2 ),
\end{equation}
and
\begin{equation}
\label{e:invIntegFactorApprox}
\dfrac{1}{\mu(\tau)}  = 1+2\tau p - ( 10\tau-2{\tau}^2 +6e^{-\tau}+e^{-2\tau}) p^2+o ( p^2)\ .
\end{equation}

On multiplying (\ref{e:diffEqnYRPS2}) by $\mu(\tau)$, we obtain
\begin{equation*}
\dfrac{d (y(\tau)\mu(\tau))}{d\tau} = \mu(\tau)\bigl( ( -2-3{e^{-\tau}} ) p + ( 16-17e^{-2\tau}+ 4e^{-\tau} + 9e^{-\tau}\tau -4e^{-2\tau}\tau) p^2 + o(p^2)\bigr)\ .
\end{equation*}

By using the approximation of the integrating factor in \eqref{e:integFactorApprox},
the above equation could be simplified to
\begin{equation*}
\dfrac{d (y(\tau)\mu(\tau))}{d\tau} = ( -3 e^{-\tau}-2 ) p+ ( 16+4\tau
+15\tau e^{-\tau}-4 \tau e^{-2 \tau} -17e^{-2 \tau}+4 e^{-\tau} ) p^2 + o(p^2)\ .
\end{equation*}

Now let $u$ be defined by
\begin{equation*}
 u = ( 3e^{-\tau}+2 )- ( 16+4 \tau
+15\tau e^{-\tau}-4\tau e^{-2 \tau} -17e^{-2 \tau}+4 e^{-\tau} )p,
\end{equation*}
such that
\begin{equation*}
\dfrac{d (y(\tau)\mu(\tau))}{d\tau} = -pu + o(p^2)\ .
\end{equation*}

Let us now find the integral of $u$ which we will need
later.
\begin{equation*}
\int_0^\tau u\, d\tau = ( -3 e^{-\tau}+2\tau) - ( 16\tau+2{\tau}^2-15
\tau e^{-\tau}-19 e^{-\tau}+2 \tau{e^{-2\tau}} +\tfrac{19}{2} e^{-2 \tau} ) p\ .
\end{equation*}

Now, suppose $u = \Omega(1)$. Then we have
\begin{equation*}
\dfrac{d (y(\tau) \mu(\tau))}{d\tau} = -(1  + o(p))pu\ .
\end{equation*}

Integrating both sides of this equation, we get
\begin{equation}
\label{e:diffEqnG2}
\begin{split}
y(\tau)\mu(\tau) = &-(1  + o(p))\bigl(( -3e^{-\tau}+2\tau )p - ( 16\tau+2{\tau}^2-15
\tau e^{-\tau}-19e^{-\tau} \\
&+2 \tau e^{-2\tau}+\tfrac{19}{2}{e^{-2\tau}}) p^2 \bigr)+ C,
\end{split}
\end{equation}

where $C$ is an arbitrary constant. We can determine the value of $C$ using
the initial condition $y(0)=1$. Thus we have
\begin{equation*}
1+7p^2+o(p^2) = 3p-\tfrac{19}{2}p^2 +C\ .
\end{equation*}

Hence, the initial condition will be satisfied if $ C = 1-3p+{\tfrac {33}{2}}p^2+o(p^2)$.
Substituting this value and (\ref{e:invIntegFactorApprox}) into (\ref{e:diffEqnG2}) and simplifying
using Taylor approximations, we get the following solution.
\begin{equation*}
y(\tau) = 1-3 ( 1 -e^{-\tau}) p +  \left( 2e^{-2\tau}\tau+\tfrac{17}{2}e^{-2\tau}-9e^{-\tau}\tau-25e^{-\tau}+
\tfrac {33}{2} \right) p^2 + o  (p^2)\ .
\end{equation*}

We can therefore conclude that while $u=\Omega(1)$, $y(\tau)$ cannot deviate much from the above solution. The solution is bounded above as required.
 It is easily verified that this solution agrees with our assumption that $y(\tau) = 1+O(p)$. Since our assumption
 is valid initially, i.e.\ $y(\tau=0) = 1$,  the solution is valid for all $\tau$. The lemma is proved.
\end{proof}
We have just proved that the generating function  $F(x,t)$ converges when $x=1+p^3$. Before looking at the subsequent
results, let us validate an assumption made in Lemma~\ref{lem:specificSolutionRPS} that $P_3^\tau = o(p^2)$.

\begin{lemma}
\label{lem:upperBoundsRPS}
The assumption that $P_3^\tau = o(p^2)$ is valid. In fact, we have
$P_\ell^\tau = o(p^2)$ for all $\ell \ge 3$.
\end{lemma}

\begin{proof}
From Lemma~\ref{lem:specificSolutionRPS} we have
\begin{equation}
\label{e:shortRunsSumRPS}
P_0^\tau + P^\tau_1 + P^\tau_2 = 1-3( 1 -e^{-\tau} ) p +  \big( 2e^{-2\tau}\tau+\tfrac {17}{2}e^{-2\tau}-9e^{-\tau}\tau -25e^{-\tau}+ \tfrac {33}{2} \big) p^2 + o ( p^2)\ .
\end{equation}

Now let $g(\tau) = F(1,t)$.  Then, from (\ref{e:diffEqnFinalRPS1}), we obtain
\begin{equation}
\label{e:diffEqnGRPS}
\dfrac{d g(\tau)}{d\tau} = p^2g(\tau)^2 + 2g(\tau)P_0^\tau p(1-p) -2p{P_0^\tau}^2(1-p)-P_0^\tau(1+3p) -P_1^\tau + 1\ .
\end{equation}

Note that, by definition,  $g(\tau)$ is equal to the sum of the probability bounds $P_{\ell}^{\tau}$. Now comparing the equations
(\ref{e:diffEqnGRPS}) and (\ref{e:diffEqnYRPS1}) reveals that both $g(\tau)$ and $y(\tau)$ are identical except some error terms
in $o(p^2)$. It is then readily verified that the solution for $g(\tau)$ will be identical to $y(\tau)$. Hence, from
Lemma~\ref{lem:generatingfunctionRPS}, we have
\begin{equation}
\label{e:diffEqnGRPS2}
g(\tau) =1-3( 1 -e^{-\tau} ) p +  \left( 2e^{-2\tau}\tau+\tfrac {17}{2} e^{-2\tau}-9e^{-\tau}\tau-25e^{-\tau}+
\tfrac {33}{2} \right) p^2 + o (p^2)\ .
\end{equation}

Notice that both (\ref{e:shortRunsSumRPS}) and (\ref{e:diffEqnGRPS2}) have the functions of the same order
on the right hand side. Hence, the additional terms that are missing on the left hand side in \eqref{e:shortRunsSumRPS} must be of the order $o(p^2)$. That is,
\begin{equation*}
\sum_{\ell \ge 3} P^\tau_\ell = o(p^2),
\end{equation*}
proving the Lemma.
\end{proof}

In Lemma~\ref{lem:generatingfunctionRPS},  we proved that the generating function $F(x,t)$ converges when $x= 1+p^3$. Hence,
if $\ell$ is sufficiently large, the following holds.
\begin{equation*}
P_\ell^\tau x^{\ell} < 1 ,  \mbox{ i.e.\ } P_{\ell}^\tau < \tfrac{1}{(1+ p^3)^{\ell}}\ .
\end{equation*}

Otherwise, there is an infinite sequence with $P_\ell^\tau \ge \tfrac{1}{x^{\ell}} $, which contributes an infinite amount to the sum,
contradicting the lemma. Thus, for some constant $\gamma > 0$, we have
\begin{equation}
\label{e:decreasingTerms}
P_{\ell}^{\tau} < \tfrac{\gamma}{(1+ p^3)^{\ell}},
\end{equation}
for all $\ell$. Using this result, the following lemma proves that it takes exponential time before a plus-run of length $\Omega(n)$ can
be formed on the cycle.

\begin{lemma}
\label{lem:nolongruns}
The following statement fails with probability exponentially small in n:
if the game is started with all minuses on the cycle, it would take exponential
time before a plus-run of length $n/4$ or longer can be created.
\end{lemma}

\begin{proof}

By definition, probability that a run of length $n/4$ starts at position $i$
at a given time $\tau$ is at most $P_{n/4}^\tau$. As the game is started with a symmetrical configuration (i.e.\ all-minuses) the
result at any time will be symmetrical too. Hence the probability that such a run exists at any position on the cycle's $n$
positions at a given time $\tau$ is equal to $nP_{n/4}^\tau$. Finally, the probability of finding such a run at any
position on the cycle at any time within $T$ steps is at most
\begin{equation*}
T n  P_{n/4}^\tau\ .
\end{equation*}

It has already been shown in \eqref{e:decreasingTerms} that, when $\ell$ is sufficiently large,
\begin{equation*}
P_{\ell} < \tfrac{\gamma}{(1+ p^3)^{\ell}}\ .
\end{equation*}

Hence the probability that a run of length $n/4$ is created in $T$ steps is at most
\begin{equation*}
T n P_{n/4}^\tau \le  T n \tfrac{\gamma}{(1+ p^3)^{n/4}}\ .
\end{equation*}
This probability is exponentially small whenever $T$ is polynomially bounded.
In other words, $T$ has to be exponentially large before a run of length $n/4$ can
appear on the cycle. Clearly, longer runs require even longer time, proving the lemma.
\end{proof}
\medskip
\begin{remark}
As mentioned earlier, the discretisation error is $O(e^{-1/p})$.
However, the analysis above has error terms $o(p^2)$. Thus, for small enough $p$,
the former is insignificant.
\end{remark}

\bigskip
\noindent \textbf{
Proof of Theorem~\ref{thm:slowConRP}: } For the game to converge to all-cooperation, at some point in time,
there must be a plus-run of length $n/4$ or longer. The result then follows from Lemma~\ref{lem:nolongruns}.

As the error in the analysis is $o(p^2)$, the value for $p$ should be small enough so that $o(p^2)$ terms
can be ignored. This completes the proof.
\hfill $\square$

\medskip
\begin{remark}
SRP also shows behaviour similar to RP for small enough $p$.
That is, we can prove that there exists a small enough $p$
for which it takes exponential time for the evolution of
cooperation for SRP. The same approach as the one used for RP
can be used here. We performed the analysis in this way and found
that it is predictably much simpler.
\end{remark}

\section{Emergence of defection}
\label{s:defection}

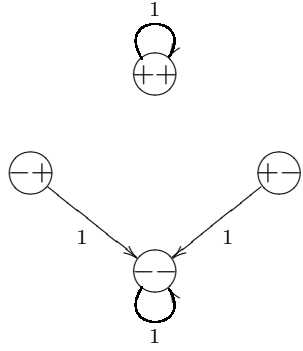
\begin{figure}
\centering
\[\xymatrix{
&*+[o][F-] {++}\ar@(ul,ur)[]^{1} & \\
*+[o][F-]{-+} \ar[dr]_{1} & & *+[o][F-]{+-}\ar[dl]^{1}  \\
&  *+[o][F-]{--} \ar@(dl,dr)[]_{1}&  } \]
\caption[RP and SRP when $p=0.$]{The transition diagram of RP and SRP when $p=0.$}
\label{fig:pavlovstrats}
\end{figure}

The case where $p=0$ is easy to analyse because $p=0$ implies that there is no randomness in the strategies.
As mentioned before, both RP and SRP are equivalent when $p=0$, thus the same analysis applies to both strategies.
The transition diagram of the resultant strategy is shown in Figure~\ref{fig:pavlovstrats}. Clearly, the process converges to
all-minuses state if there are any minuses on the cycle at the beginning of the game. Theorem~\ref{thm:defection} computes the time it takes for this.

\bigskip
\noindent \textbf{
Proof of Theorem~\ref{thm:defection}: }
It is easy to check that it will take the longest to reach the absorbing state (all-minuses state) if there is only one minus,
i.e.\ a singleton, on the cycle at the beginning of the game. Therefore, we can use this setting as the initial configuration
for the worst case analysis.

Note that at each step of the game, the probability of spreading minus to a neighbour is $2/n$. Let $T_i$ denote the number
of steps it takes to go from $i$-minuses to $(i+1)$-minuses on the cycle. Thus we have
\begin{equation*}
\Pr(T_i = t) = \left(1 - \dfrac{2}{n}\right)^{(t-1)} \dfrac{2}{n}\ .
\end{equation*}
Clearly, $T_i$ has geometric distribution with probability of success
$2/n$. Therefore $\E[T_i] = n/2$. Hence we get
\begin{equation*}
\E[T] = \sum_{i=1}^{n-1} \E[T_i] =  \sum_{i=1}^{n-1} \dfrac{n}{2} = \dfrac{n(n-1)}{2}\ .
\end{equation*}

Let us now get a bound on the probability of getting large deviations from the mean $\E[T]$. Let $X^t$ denote the event
that the number of minuses was not increased in the first $t$ trials. Then,
\begin{equation*}
\Pr(X^t) = \biggl( 1 - \frac{2}{n} \biggr)^t \le e^{\frac{-2t}{n}}\ .
\end{equation*}
If $t = \beta n \log n/2$, then
$\Pr(X^t) \le e^{\frac{- \beta n \log n}{n}} = \frac{1}{n^{\beta}}.$  But we know from the definition of $X^t$ that
\begin{equation*}
\Pr\biggl(T_i > \dfrac{\beta n \log n}{2}\biggr) \le \Pr(X^t)  \le \frac{1}{n^{\beta}}\ .
\end{equation*}
Thus, deviations of size $\frac{\beta n \log n}{2}$ are unlikely. In other words, $T_i$ lies within the range
$\left[0, \frac{\beta n \log n}{2}   \right]$ with high probability. Now, define a set of random variables $Y_i$
such that $Y_i = \frac{2 T_i}{\beta n \log n}$. Then, $Y_i \in [0, 1]$ with high probability. Also, we have
\begin{equation*}
\E[Y] = \frac{2 \E[T]}{ \beta n \log n} = \frac{n-1}{\beta \log n}\ .
\end{equation*}
As $Y_1, Y_2, \ldots , Y_n$ are independent random variables taking values in [0,1], we can apply Chernoff bound to get
\begin{equation*}
\Pr\bigl(Y \notin [(1 \pm \varepsilon) \E[Y]] \bigr) \le 2  e^{- \frac{1}{3} \varepsilon^2 \frac{n-1}{\beta \log n}}\ .
\end{equation*}
If $\varepsilon = \frac{3 \beta \log n}{\sqrt{ n-1}}$, the following holds.
\begin{equation*}
\Pr\left(Y \notin [(1 \pm \varepsilon)\E[Y]] \right) \le 2 e^{- 3 \beta \log n} = \dfrac{2}{n^{3 \beta}}\ .
\end{equation*}

It  follows immediately that $T$ lies within the range  $[(1 \pm \varepsilon) \E[T]]$ with high probability.
Thus we can conclude that $T \in \left[ \frac{n(n-1)}{2} \pm   O(n^{\frac{3}{2}} \log n) \right]$ with high
probability.
\hfill $\square$

\section{Experimental results} \label{s:simulation}

Theorem~\ref{thm:fastconvRP} proves that cooperation emerges fast when $p$ is high,
and Theorem~\ref{thm:slowConRP} shows that cooperation emerges
exponentially slowly when $p$ is small enough. As it is not clear what happens for $p$ between
these two ranges, we carried out an empirical study. The results
of this study are presented in this section.

\subsection{Simulation model}
\label{s:simulationsimulator}

The experimental results presented in this paper were obtained from  a computer program which we developed to simulate
the IPD game played by agents arranged as the vertices of a cycle. This
program takes the length of the cycle and a value for $p$ as the input parameters and plays the game until cooperation
emerges or the number of iteration reaches a predefined maximum, whichever happens first. The maximum number of iteration
attempted is $43 \times 10^6$. At each step of the game, an edge is chosen uniformly at random and the game is played by
the associated agents based on RP. Experiments were performed in a homogeneous setting where
all players on the cycle adopt the same strategy. In our experiments, the game was started with all players playing defect.

When all agents on the cycle play RP, the time taken for reaching cooperation was measured in terms of the number of steps
required and plotted against the values of $p$ in Figure~\ref{fig:timingCurves}. For the cases where the all-cooperate state was not reached in
$43 \times 10^6$ steps, the number of cooperators were counted before abandoning the game and plotted against $p$
in Figure~\ref{fig:coopFraction}.
Each data point in the graphs represents an average value of 100 repetitions.

\subsection{Observations}
\label{s:simulationobservation}

Figure~\ref{fig:timingCurves} suggests that the absorption time decreases as $p$ increases,
which is to be expected from the definition of the strategy. The results also
support our theoretical results that cooperation
emerges quite fast for high $p$, and takes a very long time for low $p$.
However there is a large gap between the minimum value of $p$ that we proved
to give fast convergence and the lowest $p$ having
relatively faster convergence. To be more precise,
Figure~\ref{fig:timingCurves} shows that the absorption time
increases rapidly when $p$ is in the region $0.5-0.6$. In
other words, the convergence is relatively much faster when $p$
is greater than $0.6$. Theorem~\ref{thm:fastconvRP}
however, rigorously proves the fast convergence for RP only
when $p\ge0.870$.

\begin{figure}
\centering
\subfigure[Timing curves for varying $n$, for higher $p$.]{
\label{fig:timingCurves}
\includegraphics[scale=0.22, angle = -90]{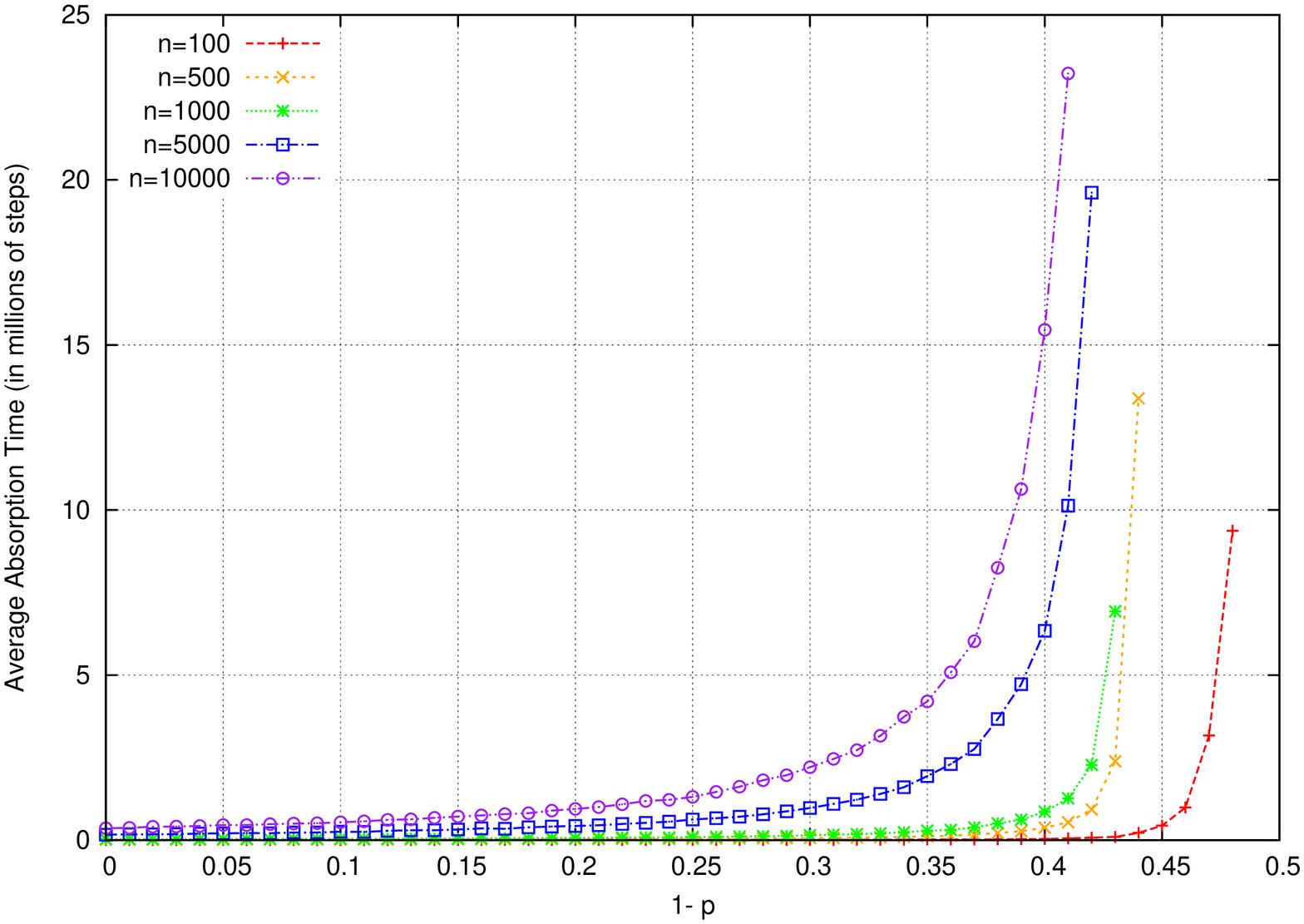}
}
\hspace{1cm}
\subfigure[Graph showing the fraction of cooperators on the
cycle after  $43 \times 10^6$ steps, for smaller $p$.]{
\label{fig:coopFraction}
\includegraphics[scale=0.22, angle = -90]{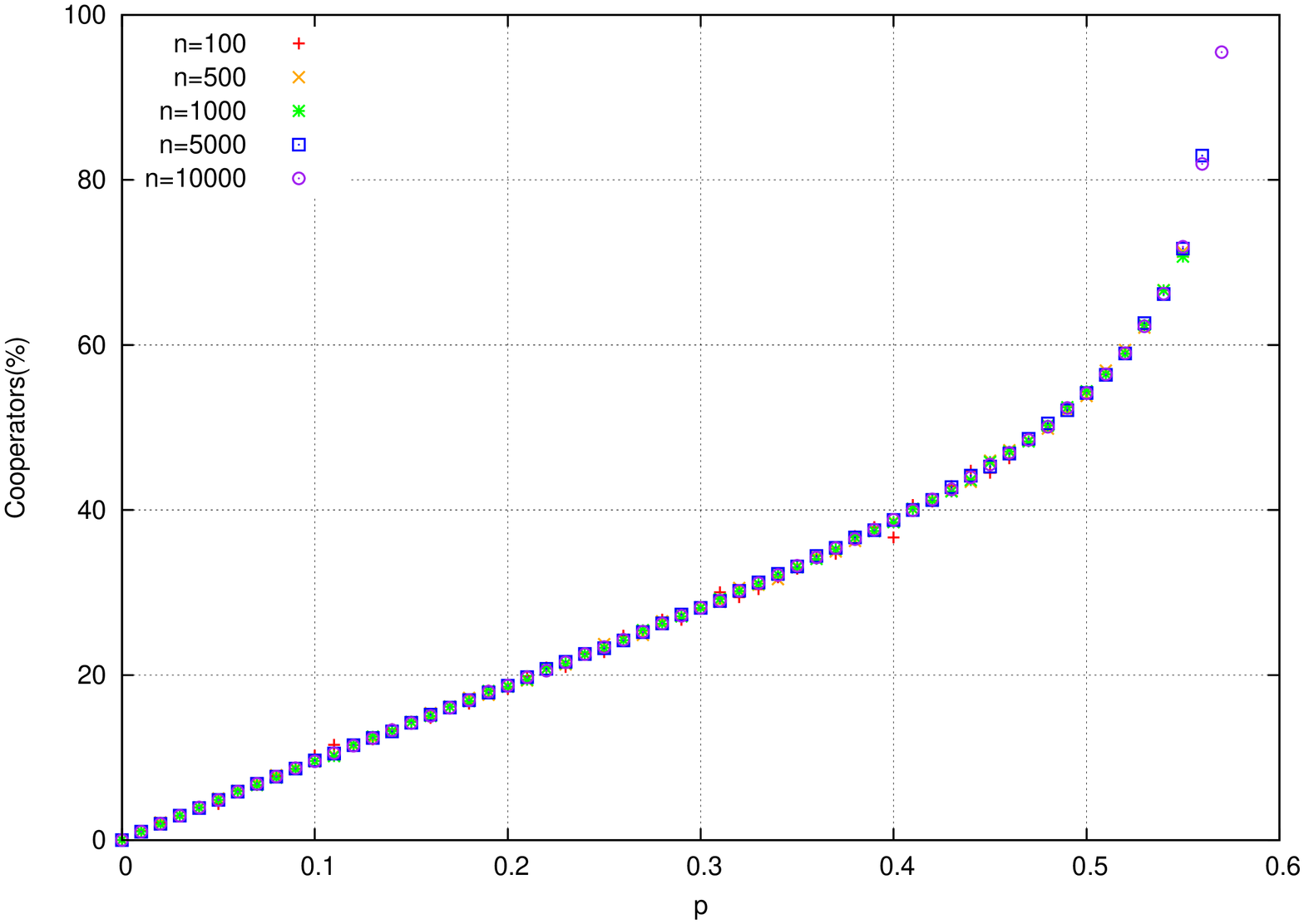}
}
\caption[Simulation results.]{Simulation results for RP when applied to the IPD on cycle with $n$ vertices.}
\label{fig:experimentalResultsRP}
\end{figure}

For small values of $p$, the emergence of cooperation took so long that we could not reliably measure the time.
This substantiates our theoretical result that it takes exponential time for cooperation to emerge for small values
of $p$. Interestingly, in Figure~\ref{fig:coopFraction}, the proportion of the cooperators is seemingly about $p$.
This can be explained intuitively as follows. When the game starts, all agents are defectors. Thereafter,
every one of them decides to cooperate with probability $p$. These are exactly the ones we will see for
smaller $p$, since their decision to cooperate will not lead to others cooperating.

In summary, the absorption time is exponentially large when $p$ is in
the region $0-0.5$. This drops considerably in the region $0.5-0.6$
and is relatively small when $p$ is greater than $0.6$.  The results suggest
that there is a sharp ``phase transition" in the region $0.5-0.6$.

\medskip
\begin{remark}
We carried out simulations for SRP as well, but the results are not included in this paper.
The results obtained are quite similar to the results presented above
for RP.
The main difference is that the apparent phase transition happens when $p$ is in
the range $0.3-0.4$ for SRP  whereas it happens in the range $0.5-0.6$ for RP.
\end{remark}

 \section{Conclusions and open problems}
\label{s:conclusion}

We have proposed randomised improvements to the Pavlov strategy for the multiplayer Iterated Prisoner's Dilemma game.
This gives two new strategies called RP (Rational Pavlov) and SRP (Simplified Rational Pavlov) with a parameter $p$.
We have studied the rate of convergence of these strategies both rigorously and experimentally
when used on the cycle for playing the IPD.
We have presented a complete analysis for RP and briefly remarked upon
similar results we obtained for SRP.

Since a rational player would choose to minimise risk without affecting long term return,
a player playing RP or SRP should choose the lowest possible $p$ that guarantees fast convergence to cooperation. Our results
provide evidence (both theoretical and empirical) that players can safely choose
$p=0.870$ for RP and $p=0.699$ for SRP, and still achieve fast cooperation. We have also shown that cooperation
emerges exponentially slow when $p$ is small enough and defection emerges
(fast) when $p = 0$, for both strategies. It is not clear what happens for intermediate $p$.
Simulation results suggest that there is a sharp phase transition in this range.

It remains as an open question whether the phase transition can be proved rigorously.
Two other interesting open questions are: whether this process can be analysed
on graphs other than cycles, and whether there are graphs with average degree greater
than 2 where fast convergence to cooperation   for RP and SRP occurs for any $p$.
\bibliographystyle{amsplain}
\bibliography{Bibs2.1}

\end{document}